\documentclass[fleqn,conference,a4paper	]{IEEEtran}

\usepackage[utf8]{inputenc}
\usepackage{amsmath}
\usepackage{amsfonts}
\usepackage{color}
\usepackage{xcolor}
\usepackage{graphicx}
\usepackage{subfigure}
\usepackage{soul}
\usepackage{amssymb}
\usepackage[outline]{contour}

\usepackage{tikz}
\usetikzlibrary{calc,trees,positioning,arrows,chains,shapes.geometric,%
    decorations.pathreplacing,decorations.pathmorphing,shapes,%
    matrix,shapes.symbols}
    \usepackage{hyperref}
\usepackage[ruled,algo2e,boxed,linesnumbered]{./algorithm2e}

 \usepackage{MnSymbol}
\makeatletter
\newcommand{\removelatexerror}{\let\@latex@error\@gobble}
\makeatother

    \newcommand{\EmBracet}[3][]{%
    \begin{tikzpicture}[overlay,remember picture]%
        \draw [decoration={brace,amplitude=0.5em},decorate,thick, #1] (#3) -- (#2);
    \end{tikzpicture}%
    }%

\newcommand{\weakestpre}[2]{\mathtt{wp}({#1},\linebreak[0] {#2})}

 \newtheorem{definition}{Definition}{\bfseries}{\itshape}
\newtheorem{theorem}{Theorem}
\newtheorem{lemma}{Lemma}
\newtheorem{property}{Property}{\bfseries}{\itshape}
\newenvironment{proof}[1][Proof]{\begin{trivlist}
\item[\hskip \labelsep {\bfseries #1}]}{\end{trivlist}}

\newcommand{\qed}{\nobreak \ifvmode \relax \else
      \ifdim\lastskip<1.5em \hskip-\lastskip
      \hskip1.5em plus0em minus0.5em \fi \nobreak
      \vrule height0.75em width0.5em depth0.25em\fi}

\usepackage[font=footnotesize]{caption}
\usepackage{framed}
\usepackage{wrapfig}

\definecolor{shadecolor}{RGB}{180,180,200}
\SetCommentSty{mycommfont}
\usetikzlibrary{fit}
\tikzset{%
  highlight/.style={rectangle,rounded corners,fill=red!15,draw,fill opacity=0.5,thick,inner sep=0pt}
}
\newcommand{\tikzmark}[2]{\tikz[overlay,remember picture,baseline=(#1.base)] \node (#1) {#2};}
\title{From Traces To Proofs: Proving Concurrent Programs Safe}

\author{\IEEEauthorblockN{Chinmay Narayan\IEEEauthorrefmark{1},
Subodh Sharma\IEEEauthorrefmark{2}, Shibashis Guha\IEEEauthorrefmark{3} and
S.Arun-Kumar\IEEEauthorrefmark{4}}
\IEEEauthorblockA{Department of Computer Science and Engineering,\\
Indian Institute of Technology Delhi\\
Email: \IEEEauthorrefmark{1}chinmay@cse.iitd.ac.in,
\IEEEauthorrefmark{2}svs@cse.iitd.ac.in,
\IEEEauthorrefmark{3}shibashis@cse.iitd.ac.in,
\IEEEauthorrefmark{4}sak@cse.iitd.ac.in}}
\begin{document}

 \newcommand{\Val}{\mathsf{Val}}
\newcommand{\Var}{\mathsf{Var}}
\newcommand{\Exp}{\mathsf{Exp}}
\newcommand{\Powerset}[1]{\mathbb{P}({#1})}
\newcommand{\Elem}[1]{\mathcal{EL}(#1)}
\newcommand{\Assign}[2]{\mathrm{#1}\mathbf{:=}\mathrm{#2}}
\newcommand{\subst}[3]{\mathtt{\linebreak[0]#1[\linebreak[0]{#2}/\linebreak[0]{#3}]}}
\newcommand{\w}[2]{!({{#1}},{#2})}
\renewcommand{\r}[2]{?({{#1}},{#2})}
\newcommand{\inarr}[1]{\begin{array}{@{}l@{}}#1\end{array}}
\newcommand{\lab}[1]{{\scriptsize{#1.}}~~}
\newcommand{\true}{\mathsf{true}}
\newcommand{\false}{\mathsf{false}}
\newcommand{\set}[1]{\{#1\}}
\newcommand{\restr}[2]{#1\downharpoonleft_{#2}}
\newcommand{\Last}{\mathbf{Last}}
\newcommand{\lland}{~\land~}
\newcommand{\llor}{~\lor~}
\newcommand{\llnot}{\lnot~}
\newcommand{\limp}{\Rightarrow}
\newcommand{\aut}[1]{\mathcal{A}(#1)}
\newcommand{\htt}[1]{\widehat{#1}}
\newcommand{\Lab}{\mathtt{Lab}}
\newcommand{\ndelta}{\mathcal{\delta}}
\newcommand{\unsatcore}{\mathtt{Unsatcore}}
\newcommand{\validcore}{\mathtt{validcore}}
\newcommand{\A}[1]{\mathcal{A}(#1)}
\newcommand{\len}[1]{|{#1}|}
\newcommand{\HOARE}[3]{\colorbox{white}{${\{{#1}\}~ \linebreak[0] #2 ~\linebreak[0]\{{#3}\}}$}}
\newcommand{\word}[1]{\mathtt{``#1"}}
\newcommand{\tafap}{B_{(\sigma, \neg \phi)}}
\newcommand{\bld}[1]{\mathbf{#1}}
\newcommand{\lang}[1]{\mathcal{L}(#1)}
\newcommand{\HMap}{\mathtt{HMap}}
\newcommand{\acc}{\mathtt{acc}}
\newcommand{\faS}{S_{\forall}}
\newcommand{\teS}{S_{\exists}}
\newcommand{\AssnMap}{\mathtt{AMap}}
\newcommand{\ResdMap}{\mathtt{RMap}}
\newcommand{\rev}{\mathsf{rev}}
\newcommand{\rf}{\mathrm{rf}}
\newcommand{\run}{\sigma}
\newcommand{\PH}{\mathsf{PH}}
\newcommand{\prog}[6]{\langle #1,\linebreak[0] #2,\linebreak[0] #3,\linebreak[0] #4, \linebreak[0]#5, \linebreak[0]#6\rangle}
\newcommand{\Proc}[1]{P_{#1}}
\newcommand{\OP}{\mathrm{op}}
\newcommand{\lcas}[3]{\mathtt{lock}({#1})}
\newcommand{\assume}[1]{\mathtt{assume}(\mathrm{#1})}
\newcommand{\assrt}[1]{\mathtt{assert}(\mathrm{#1})}
\newcommand{\SV}{\mathtt{SV}}
\newcommand{\LV}{\mathtt{LV}}
\newcommand{\p}{\mathit{p}}
\newcommand{\nop}{\mathtt{nop}}
\newcommand{\partialfun}{\hookrightarrow}
\newcommand{\AssnM}{\mathrm{Assrn}}
\newcommand{\BExp}{\mathtt{BExp}}
\newcommand{\succset}{\mathtt{succ}}
\newcommand{\track}[1]{\colorbox{yellow}{\textcolor{red}{#1}}}
\newcommand{\defeq}{\stackrel{def}{=}}
\newcommand{\afa}[1]{\mathcal{\hat{A}}_{#1}}
\newcommand{\lb}[1]{\mathtt{#1}}

 \maketitle
\thispagestyle{plain}
\begin{abstract}
  Nondeterminism in scheduling is the cardinal
  reason for difficulty in proving correctness of
  concurrent programs.  A powerful proof strategy
  was recently
  proposed~\cite{Farzan:2013:IDF:2429069.2429086}
  to show the correctness of such programs. The
  approach captured data-flow dependencies among
  the instructions of an interleaved and
  error-free execution of threads. These data-flow
  dependencies were represented by an {\em
    inductive data-flow graph} (iDFG), which, in a
  nutshell, denotes a set of executions of the
  concurrent program that gave rise to the
  discovered data-flow dependencies.  The iDFGs
  were further transformed in to alternative
  finite automatons (AFAs) in order to utilize
  efficient automata-theoretic tools to solve the
  problem.  In this paper, we give a novel and
  efficient algorithm to directly construct AFAs
  that capture the data-flow dependencies in a
  concurrent program execution.  We implemented
  the algorithm in a tool called \texttt
  {ProofTraPar} to prove the correctness of finite
  state cyclic programs under the sequentially
  consistent memory model. Our results are
  encouranging and compare favorably to existing
  state-of-the-art tools.
\end{abstract}

\tikzset{
      finalforallnode/.style={
      rectangle,
      double=white,
      double distance=2pt,
      minimum size=6mm,rounded corners=2mm,
      very thick,
      draw=red!50!black!50, 
      top color=white, 
      bottom color=red!50!black!20, 
      font=\itshape
      },
     forallnode/.style={
      rectangle,
      minimum size=6mm,rounded corners=2mm,
      very thick,
      draw=red!50!black!50, 
      top color=white, 
      bottom color=red!50!black!20, 
      font=\itshape
      },
finalexistsnode/.style={
      rectangle,
      minimum size=6mm,
      double=white,
      double distance=2pt,
      very thick,
      draw=red!50!green!50, 
      top color=white, 
      bottom color=red!50!yellow!20, 
      font=\itshape
      },
existsnode/.style={
      rectangle,
      minimum size=6mm,
      very thick,
      draw=red!50!green!50, 
      top color=white, 
      bottom color=red!50!yellow!20, 
      font=\itshape
      },
pt/.style={circle,fill=#1,inner sep=0mm,minimum size=11pt,font=\scriptsize},
}
\renewcommand{\bfdefault}{bx}
\newcommand{\fallstate}[1]{
\node[pt=black] at ({#1}.10) () {{\contourlength{0.4pt}\contournumber{1}\contour{white}{$\forall$}}};
}
\newcommand{\fexstate}[1]{
\node[pt=black] at ({#1}.10) () {{\contourlength{0.4pt}\contournumber{1}\contour{white}{$\exists$}}};
}
\newcommand{\setinitnode}[1]{\node[pt=white, left = 10pt of {#1}] (dummy) {};\draw[->] (dummy) to node [above] {} (#1);}

\newcommand{\idfg}[1]{
\scalebox{#1}{
\begin{tikzpicture}
\node[forallnode, label=left:$n_2$] (n2) {$a:y:=w$};
\node[forallnode, right= of n2, label=left:$n_3$] (n3) {$b:r:=w+1$};
\node[finalforallnode, below right= of n2, label=left:$n_1$] (n1) {$c:t:=x-1$};

\node[below=1 cm of n1] (dummy1) {};
\node[above= 1 cm of n2] (dummy2) {};
\node[above= 1 cm of n3] (dummy3) {};

\draw [->] (n1) to node [left] {$\begin{array}{l}\\ (y=w)\\ \lland (t<x) \\ \lland (r>w)\end{array}$} (dummy1) ;
\draw [->] (n2) to node [left] {$y=w$} (n1) ;
\draw [->] (n3) to node [left] {$r>w$} (n1) ;
\draw [->] (dummy2) to node [left] {$\true$} (n2) ;
\draw [->] (dummy3) to node [left] {$\true$} (n3) ;
\end{tikzpicture}
}
}

\newcommand{\idfgexafa}[1]{
\begin{tikzpicture}

\node[finalexistsnode, label=left:$s_4$,label=above:\colorbox{#1}{$\true$}] (n1) {$\true$};
\fexstate{n1};
\node[finalexistsnode, right=2cm of n1, label=above:\colorbox{#1}{$\true$},label=left:$s_5$] (n2) {$\true$};
\fexstate{n2};
\node[finalexistsnode, right=2cm of n2,label=above:\colorbox{#1}{$\true$},label=left:$s_6$] (n3) {$\true$};
\fexstate{n3};

\node[existsnode, below=1cm of n1, label=right:\colorbox{#1}{$\true$},label=left:$s_1$] (n4) {$y=w$};
\fexstate{n4};
\node[existsnode, below=1cm of n2, label=right:\colorbox{#1}{$\true$},label=left:$s_2$] (n5) {$r>w$};
\fexstate{n5};
\node[existsnode, below=1cm of n3,label=right:\colorbox{#1}{$\true$},label=left:$s_3$] (n6) {$t<x$};
\fexstate{n6};

\node[forallnode, below right=3cm of n4,label=right:\colorbox{#1}{$\true$},label=below:$s_0$] (n7) {$(y=w)\lland (r>w)\lland (t<x)$};
\fallstate{n7};
\setinitnode{n7};

\draw [->, loop left=1pt] (n1) to node [left] {$a,b,c$} (n1) ;
\draw [->,loop left=1pt] (n2) to node [left] {$a,b,c$} (n2) ;
\draw [->,loop left=1pt] (n3) to node [left] {$a,b,c$} (n3) ;

\draw [->] (n4) to node [left] {$a$} (n1) ;
\draw [->, loop below=1pt] (n4) to node [left] {$b,c$} (n4) ;
\draw [->] (n5) to node [left] {$b$} (n2) ;
\draw [->, loop below=1pt] (n5) to node [left] {$a,c$} (n5) ;
\draw [->] (n6) to node [left] {$c$} (n3) ;
\draw [->, loop below=1pt] (n6) to node [left] {$a,b$} (n6) ;
\draw [->] (n7) to node [left] {$\nop$} (n4) ;
\draw [->] (n7) to node [left] {$\nop$} (n5) ;
\draw [->] (n7) to node [left] {$\nop$} (n6) ;
\end{tikzpicture}
}

\newcommand{\dualafaex}[1]{
\scalebox{#1}{
  $\begin{array}{l}
  \lab{a} \w{Y}{w}\\
  \lab{b} \w{R}{w+1}\\
  \lab{c} \w{T}{x-1}\\
  \end{array}$
}
  }

\newcommand{\dualafa}[1]{
\begin{tikzpicture}

\node[finalexistsnode, label=right:$s_5$,label=above:\colorbox{#1}{$\true$}] (n0) {$\true$};
\fexstate{n0};
\node[finalexistsnode, below=of n0, label=above:\colorbox{#1}{$\true$},label=right:$s_6$] (n1) {$\true$};
\fexstate{n1};
\node[finalexistsnode, below=of n1, label=above:\colorbox{#1}{$\true$},label=right:$s_7$] (n2) {$\true$};
\fexstate{n2};

\node[existsnode, left=of n0,label=left:\colorbox{#1}{$\true$},label=below right:$s_2$] (n3) {$Y=w$};
\fexstate{n3};
\node[existsnode, left= of n1, label=above:\colorbox{#1}{$\true$},label=below right:$s_3$] (n4) {$R>w$};
\fexstate{n4};

\node[forallnode, left=of n4, label=above:\colorbox{#1}{$\true$},label=below right:$s_1$] (n5) {$\begin{array}{l}(Y=w)\\ \llor (R>w)\end{array}$};
\fallstate{n5};
\node[existsnode, left=of n2,label=left:\colorbox{#1}{$\true$},label=below right:$s_4$] (n6) {$T<x$};
\fexstate{n6};
\node[forallnode, left =of n5,label=above:\colorbox{#1}{$\true$},label=below:$s_0$] (n7) {$\begin{array}{l}((Y=w)\\ \llor (R>w))\\ \lland (T<x)\end{array}$};
\fallstate{n7};
\setinitnode{n7};

\draw [->, loop below=1pt] (n0) to node [above] {$a,b,c$} (n0) ;
\draw [->,loop below=1pt] (n1) to node [above] {$a,b,c$} (n1) ;
\draw [->,loop below=1pt] (n2) to node [above] {$a,b,c$} (n2) ;

\draw [->,loop above=1pt] (n3) to node [above] {$b,c$} (n3) ;
\draw [->,loop below=1pt] (n6) to node [above] {$a,b$} (n6) ;
\draw [->,loop below=1pt] (n4) to node [above] {$a,c$} (n4) ;

\draw [->] (n3) to node [above] {$a$} (n0) ;
\draw [->] (n4) to node [above] {$b$} (n1) ;
\draw [->] (n6) to node [above] {$c$} (n2) ;
\draw [->] (n5) to node [above] {$\nop$} (n3) ;
\draw [->] (n5) to node [above] {$\nop$} (n4) ;

\draw [->] (n7) to node [above] {$\nop$} (n5) ;
\draw [->] (n7) to node [above] {$\nop$} (n6) ;

\end{tikzpicture}
}

\newcommand{\petersonpro}{
\begin{tikzpicture}

\node[existsnode] (n0) {$q_a$};
\node[existsnode, below=0.5 cm of n0] (n1) {$q_b$};
\node[existsnode, below=0.5 cm of n1] (n2) {$q_c$};
\node[existsnode, below=0.5 cm of n2] (n3) {$q_d$};
\node[existsnode, below=0.5 cm of n3] (n4) {$q_e$};
\node[finalexistsnode, below=0.5 cm of n4] (n5) {$q_f$};

\node[existsnode, right=1cm of n0] (n00) {$q_p$};
\node[existsnode, below=0.5 cm of n00] (n10) {$q_q$};
\node[existsnode, below=0.5 cm of n10] (n20) {$q_r$};
\node[existsnode, below=0.5 cm of n20] (n30) {$q_s$};
\node[existsnode, below=0.5 cm of n30] (n40) {$q_t$};
\node[finalexistsnode, below=0.5 cm of n40] (n50) {$q_u$};

%
\setinitnode{n0};
\setinitnode{n00};

\draw [->, above=2pt] (n00) to node [right] {$\lb{p}$} (n10) ;
\draw [->] (n10) to node [right] {$\lb{q}$} (n20) ;
\draw [->] (n20) to node [right] {$\lb{P}$} (n30) ;
\draw [->,loop right=1pt] (n20) to node [above] {$\lb{Q}$} (n20) ;
\draw [->] (n30) to node [right] {$\lb{r}$} (n40) ;
\draw [->] (n40) to node [right] {$\lb{s}$} (n50) ;
\draw [->, bend right=40] (n50) to node [right,pos=0.8] {$\lb{t}$} (n00) ;
\draw [->, above=2pt] (n0) to node [right] {$\lb{a}$} (n1) ;
\draw [->] (n1) to node [right] {$\lb{b}$} (n2) ;
\draw [->] (n2) to node [right] {$\lb{A}$} (n3) ;
\draw [->,loop right=1pt] (n2) to node [right] {$\lb{B}$} (n2) ;
\draw [->] (n3) to node [right] {$\lb{c}$} (n4) ;
\draw [->] (n4) to node [right] {$\lb{d}$} (n5) ;
\draw [->, bend left=30] (n5) to node [left,pos=0.5] {$\lb{e}$} (n0) ;
\end{tikzpicture}
}

\newcommand{\exafa}[1]{
\scalebox{#1}{
  $\begin{array}{l}
  \lb{a.~} \Assign{Y}{x+1}\\
  \lb{b.~} \Assign{W}{t}\\
  \lb{c.~} \Assign{z}{W}\\
  \lb{d.~} \Assign{S}{t+1}\\
  \lb{e.~} \Assign{z}{Y}\\
  \end{array}$
}
  }

\newcommand{\itroneexfa}[1]{
\scalebox{#1}{
\begin{tikzpicture}
\node[forallnode,label=below:$s_0$] (nphi) {$\begin{array}{l}S<t \\ \lland  z<x\end{array}$};
\fallstate{nphi};
\setinitnode{nphi};
\end{tikzpicture}
}
}
\newcommand{\itrtwoexafa}[1]{
\scalebox{#1}{
\begin{tikzpicture}
\node[existsnode,label=below:$s_1$] (s1) {$S<t$};
\fexstate{s1};
\node[existsnode, below = of s1,label=below:$s_2$ ] (s2) {$z<x$};
\fexstate{s2};
\node[forallnode, left = 1 cm of s1,label=below:$s_0$] (nphi)  {$\begin{array}{l}S<t \\ \lland  z<x\end{array}$};
\fallstate{nphi};
\setinitnode{nphi};
\draw [->] (nphi) to node [above] {$\nop$} (s1) ;
\draw [->] (nphi) to node [above] {$\nop$} (s2) ;
\end{tikzpicture}
}
}

\newcommand{\itrthreeexafa}[1]{
\scalebox{#1}{
\begin{tikzpicture}
\node[existsnode,label=left:$s_1$] (s1) {$S<t$};
\fexstate{s1};
\node[existsnode, below = of s1,label=below:$s_2$] (s2) {$z<x$};
\fexstate{s2};
\node[forallnode, below left = 1 cm of s1,label=below:$s_0$] (nphi)  {$\begin{array}{l}S<t \\ \lland  z<x\end{array}$};
\fallstate{nphi};
\node[finalexistsnode, right = of s1,label=above:$s_3$] (s3) {$\false$};
\fexstate{s3};
\setinitnode{nphi};
\draw [->] (nphi) to node [above] {$\nop$} (s1) ;
\draw [->] (nphi) to node [above] {$\nop$} (s2) ;
\draw [->, loop below=-1pt] (s1) to node [above] {$a,b,c,e$} (s1) ;
\draw [->] (s1) to node [above] {${d}$} (s3) ;
\draw [->, loop below=-1pt] (s3) to node [above] {$a,b,c,d,e$} (s3) ;
\end{tikzpicture}
}
}
\newcommand{\itrfourexafa}[1]{
\scalebox{#1}{
\begin{tikzpicture}
\node[existsnode,label=left:$s_1$] (s1) {$S<t$};
\fexstate{s1};
\node[existsnode, below = of s1,label=left:$s_2$] (s2) {$z<x$};
\fexstate{s2};
\node[forallnode, below left = 1 cm of s1,label=below:$s_0$] (nphi)  {$\begin{array}{l}S<t \\ \lland  z<x\end{array}$};
\fallstate{nphi};
\node[finalexistsnode, right = of s1,label=above:$s_3$] (s3) {$\false$};
\fexstate{s3};
\node[existsnode, right = of s2,label=below:$s_4$] (s4) {$Y<x$};
\fexstate{s4};
\setinitnode{nphi};
\draw [->] (nphi) to node [above] {$\nop$} (s1) ;
\draw [->, loop below=-1pt] (s1) to node [above] {$a,b,c,e$} (s1) ;
\draw [->] (nphi) to node [above] {$\nop$} (s2) ;
\draw [->] (s1) to node [above] {${d}$} (s3) ;
\draw [->, loop below=-1pt] (s3) to node [above] {$a,b,c,d,e$} (s3) ;
\draw [->] (s2) to node [above] {${e}$} (s4) ;
\draw [->, loop below=2pt] (s2) to node [above] {${a,b,d}$} (s2) ;
\end{tikzpicture}
}
}
\newcommand{\itrfifthexafa}[1]{
\scalebox{#1}{
\begin{tikzpicture}
\node[existsnode,label=left:$s_1$] (s1) {$S<t$};
\fexstate{s1};
\node[existsnode, below = of s1,label=left:$s_2$] (s2) {$z<x$};
\fexstate{s2};
\node[forallnode, below left = 1 cm of s1,label=below:$s_0$] (nphi)  {$\begin{array}{l}S<t \\ \lland  z<x\end{array}$};
\fallstate{nphi};
\node[finalexistsnode, right = of s1,label=above:$s_3$] (s3) {$\false$};
\fexstate{s3};
\node[existsnode, right = of s2,label=right:$s_4$] (s4) {$Y<x$};
\fexstate{s4};
\node[finalexistsnode, below = of s2,label=below:$s_5$] (s5) {$\false$};
\fexstate{s5};
\setinitnode{nphi};
\draw [->] (nphi) to node [above] {$\nop$} (s1) ;
\draw [->] (nphi) to node [above] {$\nop$} (s2) ;
\draw [->] (s1) to node [above] {${d}$} (s3) ;
\draw [->, loop below=-1pt] (s3) to node [above] {$a,b,c,d,e$} (s3) ;
\draw [->, loop below=-1pt] (s1) to node [above] {$a,b,c,e$} (s1) ;
\draw [->] (s2) to node [above] {${e}$} (s4) ;
\draw [->, loop below=-1pt] (s2) to node [above] {${a,b,d}$} (s2) ;
\draw [->] (s4) to node [left] {${a}$} (s5) ;
\draw [->, loop below=-1pt] (s4) to node [above] {$b,c,d,e$} (s4) ;
\draw [->, loop left=-1pt] (s5) to node [above] {$a,b,c,d,e$} (s5) ;

\end{tikzpicture}
}
}

\newcommand{\exafaoneannotated}[1]{
{
\begin{tikzpicture}
\node[existsnode,label=left:$s_1$, label={[ellipse ,draw]above:{$\false$}}] (s1) {$S<t$};
\fexstate{s1};
\node[existsnode,label=below left:$s_2$, below = 2 cm of s1,label={[ellipse ,draw]above:{$\false$}}] (s2) {$z<x$};
\fexstate{s2};
\node[existsnode, label=below:$s_0$,below left = 1 cm of s1,label={[ellipse ,draw]above:{$\false$}}] (nphi)  {$\begin{array}{l}S<t \\ \lland  z<x\end{array}$};
\fallstate{nphi};
\node[finalexistsnode, label=below right:$s_3$,right = of s1,label={[ellipse ,draw]above:{$\false$}}] (s3) {$\false$};
\fexstate{s3};
\node[existsnode, label=below right:$s_4$,right = of s2,label={[ellipse ,draw]right:{$\false$}}] (s4) {$Y<x$};
\fexstate{s4};
\node[finalexistsnode, label=below right:$s_5$,below = of s2, label={[ellipse ,draw]below:{$\false$}}] (s5) {$\false$};
\fexstate{s5};
\setinitnode{nphi};
\draw [->] (nphi) to node [above] {$\epsilon$} (s1) ;
\draw [->] (nphi) to node [above] {$\epsilon$} (s2) ;
\draw [->] (s1) to node [above] {${\lb{d}}$} (s3) ;
\draw [->, loop below=-1pt] (s3) to node [below] {$\lb{a,b,c,d,e}$} (s3) ;
\draw [->, loop below=-1pt] (s1) to node [below] {$\lb{a,b,c,e}$} (s1) ;
\draw [->, loop below=-1pt] (s2) to node [below] {$\lb{a,b,d}$} (s2) ;
\draw [->] (s2) to node [above] {$\lb{e}$} (s4) ;
\draw [->] (s4) to node [left] {$\lb{a}$} (s5) ;
\draw [->, loop below=-1pt] (s4) to node [below] {$\lb{b,c,d,e}$} (s4) ;
\draw [->, loop left=-1pt] (s5) to node [left] {$\lb{a,b,c,d,e}$} (s5) ;
\end{tikzpicture}
}
}

\newcommand{\exafaoneannotatedmodified}[1]{
{
\begin{tikzpicture}
\node[existsnode, label=left:$s_1$, label={[ellipse ,draw]above:{$\false$}}] (s1) {$S<t$};
\fexstate{s1};
\node[existsnode, below = 2cm of s1,label=below left:$s_2$, label={[ellipse ,draw]above:{$\false$}}] (s2) {$z<x$};
\fexstate{s2};

 \node[existsnode, below left = 1 cm of s1,label=below:$s_0$,label={[ellipse ,draw]above:{$\false$}}] (nphi)  {$\begin{array}{l}S<t \\ \lland  z<x\end{array}$};
 \fexstate{nphi};


\node[finalexistsnode, right = of s1,label=below right:$s_3$,label={[ellipse ,draw]above:{$\false$}}] (s3) {$\false$};
\fexstate{s3};
\node[existsnode, right = of s2,label=below right:$s_4$,label={[ellipse ,draw]right:{$\false$}}] (s4) {$Y<x$};
\fexstate{s4};
\node[finalexistsnode, below = of s2, label=below right:$s_5$,label={[ellipse ,draw]below:{$\false$}}] (s5) {$\false$};
\fexstate{s5};
\setinitnode{nphi};
\draw [->] (nphi) to node [above] {$\epsilon$} (s1) ;
\draw [->] (nphi) to node [above] {$\epsilon$} (s2) ;

\draw [->] (s1) to node [above] {$\lb{d}$} (s3) ;
\draw [->, loop below=-1pt] (s3) to node [below] {$\lb{a,b,c,d,e}$} (s3) ;
\draw [->, loop below=-1pt] (s1) to node [below] {$\lb{a,b,c,e}$} (s1) ;
\draw [->] (s2) to node [above] {$\lb{e}$} (s4) ;
\draw [->] (s4) to node [left] {$\lb{a}$} (s5) ;
\draw [->, loop below=-1pt] (s4) to node [below] {$\lb{b,c,d,e}$} (s4) ;
\draw [->, loop left=-1pt] (s5) to node [left] {$\lb{a,b,c,d,e}$} (s5) ;
\draw [->, loop below=-1pt] (s2) to node [below] {$\lb{a,b,d}$} (s2) ;
\end{tikzpicture}
}
}


\newcommand{\exafatwo}[1]{
\scalebox{#1}{
  $\begin{array}{l}
  \lab{a} \r{r}{Y}\\
  \lab{b} \r{t}{Z}\\
  \end{array}$
}
  }

\newcommand{\itroneexfatwo}[1]{
\scalebox{#1}{
\begin{tikzpicture}
\node[forallnode] (nphi) {$\begin{array}{l}(x=Z \llor Y=W)  \\ \lland  t<x \\ \lland  r>W\end{array}$};
\fallstate{nphi};
\end{tikzpicture}
}
}
\newcommand{\itrtwoexafatwo}[1]{
\scalebox{#1}{
\begin{tikzpicture}
\node[forallnode] (s1) {$\begin{array}{l}x=Z \\ \llor Y=W\end{array}$};
\fallstate{s1};
\node[existsnode, below = of s1] (s2) {$t<x$};
\fexstate{s2};
\node[existsnode, below = of s2] (s3) {$r>W$};
\fexstate{s2};

\node[forallnode, left = 1cm of s2] (nphi) {$\begin{array}{l}(x=Z \llor Y=W)  \\ \lland  t<x \\ \lland  r>W\end{array}$};
\fallstate{nphi};

\draw [->] (nphi) to node [above] {$\nop$} (s1) ;
\draw [->] (nphi) to node [above] {$\nop$} (s2) ;
\draw [->] (nphi) to node [above] {$\nop$} (s3) ;
\end{tikzpicture}
}
}

\newcommand{\itrthreeexafatwo}[1]{
\scalebox{#1}{
\begin{tikzpicture}
\node[forallnode] (s1) {$\begin{array}{l}x=Z \\ \llor Y=W\end{array}$};
\fallstate{s1};
\node[existsnode, below = of s1] (s2) {$t<x$};
\fexstate{s2};
\node[existsnode, below = of s2] (s3) {$r>W$};
\fexstate{s2};

\node[forallnode, left = 1cm of s2] (nphi) {$\begin{array}{l}(x=Z \llor Y=W)  \\ \lland  t<x \\ \lland  r>W\end{array}$};
\fallstate{nphi};

\node[existsnode, right= of s1] (s4) {$x=Z$};
\fexstate{s4};
\node[existsnode, below = of s4] (s5) {$Y=W$};
\fexstate{s5};

\draw [->] (nphi) to node [above] {$\nop$} (s1) ;
\draw [->] (nphi) to node [above] {$\nop$} (s2) ;
\draw [->] (nphi) to node [above] {$\nop$} (s3) ;
\draw [->] (s1) to node [above] {$\nop$} (s4) ;
\draw [->] (s1) to node [above] {$\nop$} (s5) ;

\end{tikzpicture}
}
}

\newcommand{\itrfourexafatwo}[1]{
\scalebox{#1}{
\begin{tikzpicture}
\node[forallnode] (s1) {$\begin{array}{l}x=Z \\ \llor Y=W\end{array}$};
\fallstate{s1};
\node[existsnode, below = of s1] (s2) {$t<x$};
\fexstate{s2};
\node[existsnode, below = of s2] (s3) {$r>W$};
\fexstate{s2};

\node[forallnode, left = 1cm of s2] (nphi) {$\begin{array}{l}(x=Z \llor Y=W)  \\ \lland  t<x \\ \lland  r>W\end{array}$};
\fallstate{nphi};

\node[existsnode, right = of s1] (s5) {$Y=W$};
\fexstate{s5};
\node[existsnode, above = of s5] (s4) {$x=Z$};
\fexstate{s4};

\draw [->] (nphi) to node [above] {$\nop$} (s1) ;
\draw [->] (nphi) to node [above] {$\nop$} (s2) ;
\draw [->] (nphi) to node [above] {$\nop$} (s3) ;
\draw [->] (s1) to node [above] {$\nop$} (s4) ;
\draw [->, loop below=1pt] (s4) to node [above] {$a,b$} (s4) ;
\draw [->, loop below=1pt] (s5) to node [above] {$a,b$} (s5) ;
\draw [->] (s1) to node [above] {$\nop$} (s5) ;

\end{tikzpicture}
}
}

\newcommand{\itrfifthexafatwo}[1]{
\scalebox{#1}{
\begin{tikzpicture}
\node[forallnode] (s1) {$\begin{array}{l}x=Z \\ \llor Y=W\end{array}$};
\fallstate{s1};
\node[existsnode, below = of s1] (s2) {$t<x$};
\fexstate{s2};
\node[existsnode, below = of s2] (s3) {$r>W$};
\fexstate{s2};

\node[existsnode, right = of s2] (s6) {$Z>X$};
\fexstate{s6};

\node[forallnode, left = 1cm of s2] (nphi) {$\begin{array}{l}(x=Z \llor Y=W)  \\ \lland  t<x \\ \lland  r>W\end{array}$};
\fallstate{nphi};

\node[existsnode, right = of s1] (s5) {$Y=W$};
\fexstate{s5};
\node[existsnode, above = of s5] (s4) {$x=Z$};
\fexstate{s4};

\draw [->] (nphi) to node [above] {$\nop$} (s1) ;
\draw [->] (nphi) to node [above] {$\nop$} (s2) ;
\draw [->] (nphi) to node [above] {$\nop$} (s3) ;
\draw [->] (s1) to node [above] {$\nop$} (s4) ;
\draw [->, loop below=1pt] (s4) to node [above] {$a,b$} (s4) ;
\draw [->, loop below=1pt] (s5) to node [above] {$a,b$} (s5) ;
\draw [->] (s1) to node [above] {$\nop$} (s5) ;
\draw [->] (s2) to node [above] {$b$} (s6) ;
\draw [->, loop right=1pt] (s6) to node [above] {$a$} (s6) ;

\end{tikzpicture}
}
}

\newcommand{\itrsixthexafatwo}[1]{
\scalebox{#1}{
\begin{tikzpicture}
\node[forallnode,label=below:$s_1$] (s1) {$\begin{array}{l}x=Z \\ \llor Y=W\end{array}$};
\fallstate{s1};
\node[existsnode, below = of s1,label=below:$s_2$] (s2) {$t<x$};
\fexstate{s2};
\node[existsnode, below = of s2,label=below:$s_3$] (s3) {$r>W$};
\fexstate{s2};
\node[finalexistsnode, right = of s3,label=below:$s_7$] (s7) {$Y>W$};
\fexstate{s7};

\node[finalexistsnode, right = of s2,label=below:$s_6$] (s6) {$Z>X$};
\fexstate{s6};

\node[forallnode, left = 1cm of s2,label=below:$s_0$] (nphi) {$\begin{array}{l}(x=Z \llor Y=W)  \\ \lland  t<x \\ \lland  r>W\end{array}$};
\fallstate{nphi};
\setinitnode{nphi};
\node[finalexistsnode, right = of s1,label=right:$s_5$] (s5) {$Y=W$};
\fexstate{s5};
\node[finalexistsnode, above = of s5,label=left:$s_4$] (s4) {$x=Z$};
\fexstate{s4};

\draw [->] (nphi) to node [above] {$\nop$} (s1) ;
\draw [->] (nphi) to node [above] {$\nop$} (s2) ;
\draw [->] (nphi) to node [above] {$\nop$} (s3) ;
\draw [->] (s1) to node [above] {$\nop$} (s4) ;
\draw [->, loop below=1pt] (s4) to node [above] {$a,b$} (s4) ;
\draw [->, loop below=1pt] (s5) to node [above] {$a,b$} (s5) ;
\draw [->] (s1) to node [above] {$\nop$} (s5) ;
\draw [->] (s2) to node [above] {$b$} (s6) ;
\draw [->, loop below=-3pt] (s2) to node [above] {$a$} (s2) ;
\draw [->, loop right=1pt] (s6) to node [above] {$a,b$} (s6) ;
\draw [->] (s3) to node [above] {$a$} (s7) ;
\draw [->,loop left=1pt] (s3) to node [above] {$b$} (s3) ;
\draw [->,loop right=1pt] (s7) to node [above] {$a,b$} (s7) ;

\end{tikzpicture}
}
}

\newcommand{\groupedextwo}{
  \begin{minipage}{\linewidth}
    \scalebox{0.65}{
      \begin{minipage}{0.5\textwidth}
	\begin{minipage}{\linewidth}
	  \itroneexfatwo{1}
	  \captionof{subfigure}{AFA with initial state}
	\end{minipage}\vspace{2cm}
	\begin{minipage}{\linewidth}
	  \itrtwoexafatwo{1}
	  \captionof{subfigure}{AFA after one iteration}
	\end{minipage}
	\begin{minipage}{\linewidth}
	  \itrthreeexafatwo{1}
	  \captionof{subfigure}{AFA after one iteration}
	\end{minipage}
      \end{minipage}\hspace{4cm}
      \begin{minipage}{0.5\textwidth}
	 \begin{minipage}{\linewidth}
	    \itrfourexafatwo{1}
	    \captionof{subfigure}{AFA after three iterations}
	  \end{minipage}
	  \begin{minipage}{\linewidth}
	    \itrfifthexafatwo{1}
	    \captionof{subfigure}{AFA after two iterations}
	   \end{minipage}
	  \begin{minipage}{\linewidth}
	    \itrsixthexafatwo{1}
	    \captionof{subfigure}{AFA after two iterations}
	   \end{minipage}
      \end{minipage}
    }
  \captionof{figure}{Construction of an AFA for example program \ref{fig:exafaconstruction2}(a) and $\neg \phi=S<t \lland z<x$}
  \label{fig:exafaconstruction2}
  \end{minipage}
}

\newcommand{\exafatwoannotated}[1]{
 {
\begin{tikzpicture}
\node[forallnode,label=left:$s_1$, label=above:\colorbox{#1}{$\begin{array}{l}x=Z \\ \llor Y=W\end{array}$}] (s1) {$\begin{array}{l}x=Z \\ \llor Y=W\end{array}$};
\fallstate{s1};
\node[existsnode, below = of s1, label=below left:$s_2$, label=below:\colorbox{#1}{$Z>X$}] (s2) {$t<x$};
\fexstate{s2};
\node[existsnode, below = of s2, label=below left:$s_3$,label=below:\colorbox{#1}{$Y>W$}] (s3) {$r>W$};
\fexstate{s2};
\node[finalexistsnode, right = of s3,label=below left:$s_7$, label=below:\colorbox{#1}{$Y>W$}] (s7) {$Y>W$};
\fexstate{s7};

\node[finalexistsnode, right = of s2, label=below left:$s_6$,label=below:\colorbox{#1}{$Z>X$}] (s6) {$Z>X$};
\fexstate{s6};

\node[forallnode, left = 1cm of s2, label=above:$s_0$, label=below:\colorbox{#1}{$\begin{array}{l}x=Z  \llor Y=W\\ \lland Z>X \\ \lland Y>W\end{array}$}] (nphi) {$\begin{array}{l}(x=Z \llor Y=W)  \\ \lland  t<x \\ \lland  r>W\end{array}$};
\fallstate{nphi};
\setinitnode{nphi};

\node[finalexistsnode, right = of s1,label=below left:$s_5$,label=right:\colorbox{#1}{$Y=W$}] (s5) {$Y=W$};
\fexstate{s5};
\node[finalexistsnode, above = of s5,label=left:$s_4$,label=right:\colorbox{#1}{$x=Z$}] (s4) {$x=Z$};
\fexstate{s4};

\draw [->] (nphi) to node [above] {$\nop$} (s1) ;
\draw [->] (nphi) to node [above] {$\nop$} (s2) ;
\draw [->] (nphi) to node [above] {$\nop$} (s3) ;
\draw [->] (s1) to node [above] {$\nop$} (s4) ;
\draw [->, loop below=1pt] (s4) to node [above] {$a,b$} (s4) ;
\draw [->, loop below=1pt] (s5) to node [above] {$a,b$} (s5) ;
\draw [->] (s1) to node [above] {$\nop$} (s5) ;
\draw [->] (s2) to node [above] {$b$} (s6) ;
\draw [->, loop right=1pt] (s6) to node [above] {$a,b$} (s6) ;
\draw [->, loop above=1pt] (s2) to node [above] {$a$} (s2) ;
\draw [->, loop left=1cm] (s3) to node [above] {$b$} (s3) ;
\draw [->] (s3) to node [above] {$a$} (s7) ;
\draw [->, loop right=1pt] (s7) to node [above] {$a,b$} (s7) ;

\end{tikzpicture}
}
}

\newcommand{\comparisontraceold}
{
  \scriptsize
  $\begin{array}{l}
  \lab{a} x:=y+1\\
  \lab{b} z:=w+2\\
  \lab{c} t:=y\\
  \end{array}$
  \normalsize
}

\newcommand{\comparisontrace}
{
  \scriptsize
  $\begin{array}{l}
  \lab{a} y:=w\\
  \lab{b} r:=w+1\\
  \lab{c} t:=x-1\\
  \end{array}$
  \normalsize
}

\newcommand{\comparisonidfg}
{
\begin{tikzpicture}
\node[existsnode, label=right:$s_1$] (s1) {$r:=w+1$};
\node[existsnode, above right= 0.8 cm of s1, label=right:$s_0$] (s0) {$init$};
\node[existsnode, below right= 0.8 cm of s0, label=right:$s_2$] (s2) {$y:=w$};
\node[existsnode,below right  = 0.5 cm of s1, label=right:$s_3$] (s3) {$t:=x-1$};
\node[circle,above=of s0] (in) {$d$};
\node[circle,below=of s3] (out) {$d$};

\draw[->] (s0) to node [left] {$\small{\true}$} (s1);
 \draw[->] (s0) to node [right] {$w>3$} (s2);
\draw[->] (s1) to node [left] {$r>w$} (s3);
 \draw[->] (s2) to node [right] {$y>3$} (s3);
 \draw[->] (in) to node [right] {$w>3$} (s0);
 \draw[->] (s3) to node [right] {$\begin{array}{l}y>3 \lland t<x \\ \lland r>w\end{array}$} (out);
\end{tikzpicture}
}

\newcommand{\comparisonidfgold}
{
\begin{tikzpicture}
\node[existsnode,label=right:$s_0$] (s0) {$init$};
\node[existsnode,below left= 0.5 cm of s0, label=right:$s_1$] (s1) {$x:=y+1$};
\node[existsnode,below right= 0.5cm of s0, label=right:$s_2$] (s2) {$z:=w+2$};
\node[existsnode,below right= 0.5 cm of s1, label=right:$s_3$] (s3) {$t:=y$};
\node[circle,above=of s0] (in) {$d$};
\node[circle,below=of s3] (out) {$d$};

\draw[->] (s0) to node [left] {$\small{\true}$} (s1);
\draw[->] (s0) to node [right] {$w>3$} (s2);
\draw[->] (s1) to node [left] {$x=y+1$} (s3);
\draw[->] (s2) to node [right] {$z>5$} (s3);
\draw[->] (in) to node [right] {$w>3$} (s0);
\draw[->] (s3) to node [right] {$\begin{array}{l}x=t+1 \\ \lland z>5\end{array}$} (out);
\end{tikzpicture}
}

\newcommand{\comparisonAFA}
{
\begin{tikzpicture}
\node[forallnode,label=right:$q_0$] (q0) {$(\begin{array}{l}x=t+1\\ \lland z>5\end{array}),s_3$};
\node[existsnode,below left= 0.5 cm of q0, label=below right:$q_1$] (q1) {$(x=y+1),s_1$};
\node[existsnode,below right= 0.5 of q0, label=below right:$q_2$] (q2) {$(z>5),s_2$};
\node[finalexistsnode,below = of q1, label=below:$q_3$] (q3) {$(\true),s_0$};
\node[finalexistsnode,below = of q2, label=below:$q_4$] (q4) {$(w>3),s_0$};
\setinitnode{q0};
\fallstate{q0};
\fexstate{q1};
\fexstate{q2};
\fexstate{q3};
\fexstate{q4};
\draw[->] (q0) to node [left] {$c$} (q1);
\draw [->, loop left=1cm] (q1) to node [right] {$b,c$} (q1) ;
\draw[->] (q0) to node [right] {$c$} (q2);
\draw [->, loop left=1cm] (q2) to node [right] {$a,c$} (q2) ;
\draw[->] (q1) to node [left]  {$a$} (q3);
\draw [->, loop left=1cm] (q3) to node [right] {$a,b,c$} (q3) ;
\draw[->] (q2) to node [right] {$b$} (q4);
\draw [->, loop left=1cm] (q4) to node [right] {$a,b,c$} (q4) ;
\end{tikzpicture}
}

\newcommand{\comparisonourAFA}{
\begin{tikzpicture}
\node[forallnode,label=below:$q_0$, label={[ellipse,draw]right:$w>3$}] (q0) {$\begin{array}{l}x=t+1\\ \land z>5\end{array}$};
\node[existsnode,below left= of q0, label=below right:$q_1$,label={[ellipse,draw]right:$w>3$}] (q1) {$z>5$};
\node[existsnode,below right= of q0, label=below right:$q_2$, label={[ellipse,draw]right:$\true$}] (q2) {$x=t+1$};
\node[finalexistsnode,below = of q1, label=below:$q_3$, label={[ellipse,draw]right:$w>3$}] (q3) {$w>3$};
\node[existsnode,below = of q2, label=below right:$q_4$, label={[ellipse,draw]right:$\true$}] (q4) {$x=y+1$};
\node[finalexistsnode,below = of q4, label=below:$q_5$, label={[ellipse,draw]right:$\true$}] (q5) {$\small{\true}$};
\setinitnode{q0};
\fallstate{q0};
\fexstate{q1};
\fexstate{q2};
\fexstate{q3};
\fexstate{q4};
\fexstate{q5};
\draw[->] (q0) to node [left] {$\epsilon$} (q1);
\draw[->] (q0) to node [left] {$\epsilon$} (q2);
\draw[->] (q1) to node [right] {$b$} (q3);

\draw [->, loop left=1cm] (q3) to node [above] {$a,b,c$} (q3) ;
\draw [->, loop left=1cm] (q1) to node [above] {$a,c$} (q1) ;
\draw [->, loop left=1cm] (q2) to node [right] {$b$} (q2) ;
\draw[->] (q2) to node [left]  {$c$} (q4);
\draw [->, loop left=1cm] (q4) to node [above] {$b,c$} (q4) ;
\draw[->] (q4) to node [left]  {$a$} (q5);
\draw [->, loop left=1cm] (q5) to node [above] {$a,b,c$} (q5) ;
\end{tikzpicture}
}

\section{Introduction}\label{sec:intro}
The problem of checking whether or not a correctness
property (specification) is violated by the program
(implementation) is already known to be challenging in a
sequential set-up, let alone when programs are implemented
exploiting concurrency.
The central reason for greater complexity in verification of
concurrent implementations is due to the exponential
increase in the number of executions.  A concurrent program
with $n$ threads and $k$ instructions per thread can have
$(nk)!/(k!)^n$ executions under a sequentially consistent
(SC)\cite{Lamport:1979:MMC:1311099.1311750} memory model.  A
common approach to address the complexity due to the
exponential number of executions is {\em trace
  partitioning}.

In \cite{Farzan:2013:IDF:2429069.2429086}, a powerful proof
strategy was presented which utilized the notion of trace
partitioning.
\renewcommand{\bld}[1]{\mathtt{#1}} Let us take Peterson's
algorithm in Figure \ref{fig:petersons} to convey the
central idea behind the trace partitioning approach.  In this
algorithm, two processes, $P_i$ and $P_j$, coordinate to
achieve an exclusive access to a critical section (CS) using
shared variables.
A process $P_i$ will busy-wait if $P_j$ has expressed
interest to enter its CS and $\bld{t}$ is $j$.
%

In order to prove the mutual exclusion (ME) property of
Peterson's algorithm, we must consider the boolean
conditions of the while loops at control locations 3 and 8.
the ME property is established only when at most one of these
conditions is false under every execution of the program,
{\em i.e.\xspace}, ME must be shown to hold true on
unbounded number of traces (trace is ``a sequence of events
corresponding to an interleaved execution of processes in
the program''\cite{Gupta:2015:SRC:2775051.2677008})
generated due to unbounded number of unfoldings of the
loops. Notice that events at control locations 3 and 8 are
{\em data-dependent} on events from control locations
$2,6,7$ and $1,2,7$, respectively.  In any finite prefix of
a trace of $P_i\| P_j$ (interleaved execution of $ P_i$ and
$P_j$) up to the events corresponding to control location 3
or 8, the last instance of event at control location 2,
$\bld{lst2}$, and the last instance of event at control
location 7, $\bld{lst7}$, can be ordered in only one of the
following two ways; either $\bld{lst2}$ appears before
$\bld{lst7}$ or $\bld{lst2}$ appears after $\bld{lst7}$.
This has resulted in partitioning of an unbounded set of
traces to a set with mere two traces.
\begin{figure}[!t]
\centering{\footnotesize $\mathrm{flag_i}=false,\mathrm{flag_j}=false,\mathrm{t}=0$; \normalsize}\\
 $\begin{array}{l@{~~}|@{~~}l}
\footnotesize{
\inarr{
{}\\ ~~~~~~ P_{i}~~~~~\\
While(true)\{\\
\lab{1} \Assign{flag_i}{true};\\
\lab{2} \Assign{t}{j};\\
\lab{3} \mathrm{while}(\mathrm{flag}_j=\mathrm{true} \,\&\, \mathrm{t}=j);\\
\lab{4}  \bld{ //Critical~Section}\\
\lab{5} \Assign{flag_i}{false};\\
\}
}
}&
\footnotesize{
\inarr{
{}\\ ~~~~~~ P_{j}~~~~~\\
While(true)\{\\
\lab{6} \Assign{flag_j}{true};\\
\lab{7} \Assign{t}{i};\\
\lab{8} \mathrm{while}(\mathrm{flag_i}=\mathrm{true}  \,\&\, \mathrm{t}= i);\\
\lab{9} \bld{//Critical~Section}\\
\lab{10} \Assign{flag_j}{false}; \\
\}
}
}
\end{array}$
\caption {Peterson's algorithm for two processes $\Proc{i}$ and $\Proc{j}$}
\label{fig:petersons}
\vspace{-0.5em}
\end{figure}

When $\bld{lst2}$ appears before $\bld{lst7}$, then the
final value of the variable $\bld{t}$ is $i$, thus making
the condition at control location 8 to be $\true$.  In the
other case, when $\bld{lst2}$ appears after $\bld{lst7}$,
the final value of the variable $\bld{t}$ is $j$, thereby
making the condition at control location 3 evaluate to
$\true$.  Hence, in no trace both the conditions are false
simultaneously. This informal reasoning indicates that both
processes can never simultaneously enter in their critical
sections.
Thus, proof of correctness for Peterson's algorithm can be
demonstrated by picking two traces, as mentioned above, from
the set of infinite traces
and proving them correct.  In general, the intuition is that
a proof for a single trace of a program can result in
pruning of a large set of traces from consideration. To
convert this intuition to a feasible verification method,
there is a need to construct a formal structure from a proof
of a trace $\sigma$ such that the semantics of this
structure includes a set of all those traces that have proof
arguments equivalent to proof of $\sigma$.
%
\emph{Inductive Data Flow Graphs} (iDFG) was proposed in
\cite{Farzan:2013:IDF:2429069.2429086} to capture
data-dependencies among the events of a trace and to perform
trace partitioning. All traces that have the same iDFG must
have the same proof of correctness. In every iteration of
their approach, a trace is picked from the set of \emph{all}
traces that is yet to be covered by the \emph{iDFG}. An iDFG
is constructed from its proof. The process is repeated until
all the traces are either covered in the \emph{iDFG} or a
counter-example is found.  An intervening step is involved
where the iDFG is converted to an {\em alternating finite
  automaton} (AFA).  While we explain AFA in later sections,
it suffices to understand at this stage that the language
accepted by this AFA and the set of traces captured by the
corresponding iDFG is the same. Their reason for this
conversion is to leverage the use of automata-theoretic
operations such as subtraction, complement {\em
  etc\xspace.}, on the set of traces.
Though the goal of paper
\cite{Farzan:2013:IDF:2429069.2429086} is verification of
concurrent programs which is the same as in this work, our
work has some crucial differences: (i) An AFA is constructed
directly from the proof of a trace without requiring the
iDFG construction, (ii) the verification procedure built on
directly constructed AFA is shown to be sound and complete
(weakest-preconditions are used to obtain the proof of
correctness of a trace), (iii) to the best of our knowledge,
we provide the first implementation of the proof strategy
discussed in \cite{Farzan:2013:IDF:2429069.2429086}.
\begin{figure}
 \removelatexerror
  \begin{minipage}[b]{\linewidth}
	 \begin{minipage}[b]{0.2\linewidth}
	    \scalebox{1}{\parbox{\linewidth}{%
	    \centering{
	    \comparisontrace}
	  }}
	  \captionof{subfigure}{}
	  \end{minipage}
\begin{minipage}[b]{0.4\linewidth}
	\begin{minipage}{\linewidth}
	    \scalebox{0.8}{\parbox{\linewidth}{%
	    \centering{
	      $~~~~~~~~~\set{abc,bac}$}}}
	    \captionof{subfigure}{}
	    \label{dd}
	  \end{minipage}
	\begin{minipage}{\linewidth}
	    \scalebox{0.8}{\parbox{\linewidth}{%
	    \centering{
	      $~~~~~\set{abc,bac,acb,cab,bca,cba}$}}}
	    \captionof{subfigure}{}
	    \label{dd}
	  \end{minipage}
\end{minipage}
\hspace{-0.5cm}\begin{minipage}[b]{0.4\linewidth}
	    \scalebox{0.6}{\parbox{\linewidth}{%
	    \centering{
	       	    \comparisonidfg}}}
	      \captionof{subfigure}{}
	      \label{dde}
	  \end{minipage}
      \centering{
      \caption{Comparison with \cite{Farzan:2013:IDF:2429069.2429086}}
      \label{fig:examplecomparison}}
      \vspace{-0.6cm}
 \end{minipage}
\end{figure}
The example trace of Figure \ref{fig:examplecomparison}(a)
highlights the key difference between iDFG to AFA conversion
of \cite{Farzan:2013:IDF:2429069.2429086} and the direct
approach presented in this work. Note that all three events
$a,b$, and $c$ are data independent, hence every resulting
trace after permuting the events in $abc$ also satisfies the
same set of pre- and post-conditions. For a Hoare triple
$\HOARE{w>3}{abc}{y>3 \lland t<x \lland r>w}$, Figure
\ref{fig:examplecomparison}(b) shows the set of traces
admitted by an AFA (obtained from iDFG shown in
Figure~\ref{fig:examplecomparison}(d)) after the first
iteration, as computed by
\cite{Farzan:2013:IDF:2429069.2429086}.  This set clearly
does not represent every permutation of $abc$; consequently,
more iterations are required to converge to an AFA that
represents all traces admissible under the same set of pre-
and post-conditions. In contrast, the AFA that is
constructed directly by our approach from the Hoare triple
$\HOARE{w>3}{abc}{y>3 \lland t<x \lland r>w}$, admits the
set of traces shown in Figure
\ref{fig:examplecomparison}(c).  Hence, on this example, our
strategy terminates in a single iteration.

To summarize, the contributions of this work are as follows:
\begin {itemize}
\item we present a novel algorithm to directly construct an
  AFA from a proof of a sequential trace of a finite state
  (possibly cyclic) concurrent program.  This construction
  is used to give a sound and complete verification
  procedure along the lines of
  \cite{Farzan:2013:IDF:2429069.2429086}. 
%
\item While \cite{Farzan:2013:IDF:2429069.2429086} allowed
  the use of any sequential verification method to construct
  a proof of a given trace, the paper does not comment on
  the performance and the feasibility of their approach due
  to the lack of an implementation. The second contribution
  of this paper is an implementation in the form of a tool,
  \texttt{ProofTraPar}. We compare our implementation
  against other state-of-the-art tools in this domain, such
  as THREADER \cite{Gupta:2011:TCV:2032305.2032337} and
  Lazy-CSeq \cite{Lazy-CSeq} (winners in the concurrency
  category of the software verification competitions held in
  2013, 2014, and 2015). \texttt{ProofTraPar}, on 
  average, performed an order of magnitude better than
  THREADER and 3 times better than Lazy-CSeq.
\end{itemize}
The paper is organized as follows: {Section
  {\ref{sec:prelim}}} covers the notations, definitions and
programming model used in this paper; Section
{\ref{sec:example}} presents our approach with the help of
an example to convey the overall idea and describes in
detail the algorithms for constructing the proposed
alternating finite automaton along with their correctness
proofs. This section ends with the overall verification
algorithm with the proof of its soundness and completeness
for finite state concurrent programs.  Section
{\ref{sec:experiments}} presents the experimental results
and comparison with existing tools namely THREADER
\cite{Gupta:2011:TCV:2032305.2032337} and Lazy-CSeq
\cite{Lazy-CSeq}.  Section {\ref{sec:related}} presents the
related work and Section {\ref{sec:conclude}} concludes with
possible future directions.

\section{Preliminaries} \label{sec:prelim}
\subsection{Program Model}\label{subsec:pro} We consider
shared-memory concurrent programs composed of a fixed number
of deterministic sequential processes and a finite set of
shared variables $\SV$. A \emph{concurrent program} is a
quadruple $\mathcal{P}=(P,A,\mathcal{I},\mathcal{D})$ where
$P$ is a finite set of processes, $A=\set{A_p\mid p\in P}$
is a set of automata, one for each process specifying their
behaviour, $\mathcal{D}$ is a finite set of constants
appearing in the syntax of processes and $\mathcal{I}$ is a
function from variables to their initial values. Each
process $p\in P$ has a disjoint set of local variables
$\LV_p$. Let $\Exp_p$ ($\BExp_p$) denote the set of
expressions (boolean expressions), ranged over by
$\mathrm{exp}$ ($\phi$) and constructed using shared
variables, local variables, $\mathcal{D}$, and standard
mathematical operators. Each specification automaton $A_p$
is a quadruple $\langle Q_p, q_p^{init}, \delta_p,
\AssnM_p\rangle$ where $Q_p$ is a finite set of control
states, $q_p^{init}$ is the initial state, and
$\AssnM_p\subseteq Q_p \times \BExp_p$ is a relation
specifying the assertions that must hold at some control
state. Each transition in $\delta_p$ is of the form
$(q,op_p,q')$ where $op_p\in
\set{\Assign{x}{exp},\assume{\phi},\lcas{x}{v_1}{v_2}}$. Here
$\Assign{x}{exp}$ evaluates $\mathrm{exp}$ in the current
state and assigns the value to $\mathrm{x}$ where
$\mathrm{x}\in \SV \cup \LV_p$. $\assume{\phi}$ is a
blocking operation that suspends the execution if the
boolean expression $\phi$ evaluates to $\false$ otherwise it
acts as $\nop$. This instruction is used to encode control
path conditions of a program. $\lcas{x}{v_1}{v_2}$, where
$x\in\SV$, is a blocking operation that suspends the
execution if the value of $x$ is not equal to $0$ otherwise
it assigns $1$ to $x$. Operation unlock is achieved by
assigning $0$ to this shared variable. Each of these
operations are deterministic in nature, i.e. execution of
any two same operations from the same states always give the
same behaviour. In all examples of this paper, we use
symbolic labels to succinctly represent program
operations. For example, Figure \ref{fig:petersonpro} shows
the specification of two processes in Peterson's
algorithm. Labels $\{\lb{a,b,p,q}\cdots\}$ denote operations
in the program. Variable $\mathrm{res}$ is introduced to
specify the mutual exclusion property as a safety
property. A process $P_i$ sets this variable to $i$ inside
its critical section. Assertions $\assrt{res=i}$ is checked
in $P_i$ before leaving its critical section. If these
assertions hold in every execution of these two processes
then the mutual exclusion property holds. These assertions
are shown as $\AssnM_{P_1}(q_f)$ and $\AssnM_{P_2}(q_u)$ in
Figure \ref{fig:petersonpro} and they need to be checked at
state $q_f$ and $q_u$ respectively.
\begin{figure}
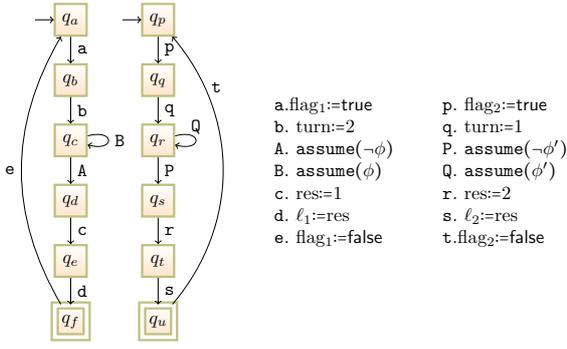

\removelatexerror
 \begin{minipage}[b]{\linewidth}
      \begin{minipage}[b]{0.3\linewidth}
	  \scalebox{0.7}{\parbox{\linewidth}{%
	    \petersonpro}} 
      \end{minipage}
      \begin{minipage}[b]{0.2\linewidth}
      \hspace{0.4cm}\scalebox{0.7}{\parbox{\linewidth}{%
	$\begin{array}{l}
	\lb{a}.\Assign{flag_1}{\true}\\
	\lb{b.~} \Assign{turn}{2}\\
	\lb{A.~} \assume{\neg \phi}\\
	\lb{B.~} \assume{\phi}\\
	\lb{c.~} \Assign{res}{1}\\
	\lb{d.~} \Assign{\ell_1}{res}\\
	\lb{e.~} \Assign{flag_1}{\false}
	\end{array} $
	}}
      \end{minipage}
      \begin{minipage}[b]{0.4\linewidth}
      \hspace{0.4cm}\scalebox{0.7}{\parbox{\linewidth}{%
	$\begin{array}{l}
	\lb{p.~} \Assign{flag_2}{\true}\\
	\lb{q.~} \Assign{turn}{1}\\
	\lb{P.~} \assume{\neg \phi'}\\
	\lb{Q.~} \assume{\phi'}\\
	\lb{r.~} \Assign{res}{2}\\
	\lb{s.~} \Assign{\ell_2}{res}\\
	\lb{t.} \Assign{flag_2}{\false}
	\end{array}$}}
      \end{minipage}
      \scalebox{0.7}{\parbox{\linewidth}{%
      \begin{center}$\begin{array}{l}\AssnM_{P_1}(q_f)\defeq\mathrm{(\ell_1=1)},~\AssnM_{P_2}(q_u)\defeq\mathrm{(\ell_2=2)}\\
                 \phi\defeq \mathrm{flag_2=\true~ \&\&~ turn=2}, \phi'\defeq \mathrm{flag_1=\true~ \&\&~ turn=1}
	  \end{array}$\end{center}
      }}
 \caption{Specification of Peterson's algorithm}
 \vspace{-0.5cm}
  \label{fig:petersonpro}
 \end{minipage}
\end{figure}
A tuple, say $t$, of $n$ elements can be represented as a function such that $t[k]$ returns the $k^{th}$ element of this tuple. Given a function $fun$, $fun[a\leftarrow b]$ denotes another function same as of $fun$ except at $a$ where it returns $b$. 

\paragraph{Parallel Composition in the SC memory model} Given a concurrent program $\mathcal{P}=(P,A,\mathcal{I},\mathcal{D})$ consisting of $n$ processes $P=\set{p_1,\cdots,p_n}$ we define an automaton $\aut{\mathcal{P}}=(\overline{Q},\overline{q}^{init},\overline{\delta},\overline{\AssnM})$ to represent the parallel composition of $\mathcal{P}$ in the SC memory model. Here $\overline{Q}=Q_{p_1}\times \cdots \times Q_{p_n}$ is the set of states ranged over by $\overline{q}$,
$\overline{q}^{init}=(q_{p_1}^{init},\cdots,q_{p_n}^{init})$ is the initial state, and
transition relation $\overline{\delta}$ models the interleaving semantics. Formally, $(\overline{q},op_j,\overline{q}')\in \overline{\delta}$ iff there exists a $j\in \set{1\cdots n}$ such that $\overline{q}[j]=q_{p_j}$, $\overline{q}'=\overline{q}[j\leftarrow q'_{p_j}]$ and $(q_{p_j},op_j,q'_{p_j})\in \delta_{p_j}$.
For a state $\overline{q}$, let $T(\overline{q})=\set{\AssnM_{p_i}(\overline{q}[i])\mid i\in \set{1\cdots n}}$. If $T(\overline{q})$ is not empty then $\overline{\AssnM}(\overline{q})$ is the conjunction of assertions in the set $T(\overline{q})$. Relation $\overline{\AssnM}$ captures the assertions which need to be checked in the interleaved traces of $P$. As our interest lies in analyzing those traces which reach those control points where assertions are specified, we mark all those states where the relation $\overline{\AssnM}$ is defined as accepting states. Every word accepted by $\aut{\mathcal{P}}$ represents one SC execution leading to a control location where at least one assertion is to be checked.
%
\subsection{Weakest Precondition} 
Given an operation $op\in \mathcal{OP}(P)$ and a postcondition formula $\phi$, the weakest precondition \emph{of $op$ with respect to $\phi$}, denoted by $\weakestpre{op}{\phi}$, is the weakest formula $\psi$ such that,
 starting from any program state $s$ that satisfies $\psi$, the execution of the operation $op$ terminates and the resulting program state $s'$ satisfies $\phi$.
 \begin{figure}
 \removelatexerror
 \begin{minipage}[b]{\linewidth}
 \scalebox{0.65}{\parbox{\linewidth}{%
    \[
      \weakestpre{op}{\phi}= 
  \begin{cases}
       \phi & \text{if }op\defeq \mathrm{skip} \\
      \subst{\phi}{x}{exp} &  \text{if }op\defeq \Assign{x}{exp}\\
      \phi \lland \phi' & \text{if }op\defeq \assrt{\phi'}\\
      \weakestpre{\assume{x=0}}{\weakestpre{\Assign{x}{1}}{\phi}} & \text{if } op\defeq\lcas{x}{v_1}{v_2}\\
      \weakestpre{op_1}{\weakestpre{op_2}{\phi}} & \text{if } op\defeq op_1.op_2\\
      \phi'\limp \phi& \text{if } op\defeq\assume{\phi'}
  \end{cases}
  \]
  }}
  \caption{Weakest precondition axioms}
  \label{fig:wpstmts}
   \vspace{-0.5cm}
 \end{minipage}
 \end{figure}

%
%
Given a formula $\phi$, variable $X$ and expression $e$, let $\subst{\phi}{X}{e}$ denote 
the formula obtained after substituting all free occurrences of $X$ by $e$ in $\phi$. 
We assume an equality operator over formulae that represents syntactic equality. Every formula is assumed to be normalized in a conjunctive normal form (CNF). We use $\true$ ($\false$) to syntactically represent a logically valid (unsatisfiable) formula.
Weakest precondition axioms for different program statements are shown in Figure \ref{fig:wpstmts}. Here empty sequence of statements is denote by $\mathrm{skip}$.
We have the following properties about weakest preconditions.
\begin{property}\label{prop:wp}
If $\weakestpre{op}{\phi_1}=\psi_1$ and $\weakestpre{op}{\phi_2}=\psi_2$ then,
\begin{itemize}
\item $\weakestpre{op}{\phi_1\lland \phi_2}=\psi_1\lland \psi_2$, and
\item $\weakestpre{op}{\phi_1\llor \phi_2}=\psi_1\llor \psi_2$. Note that this property holds only when $S$ is a deterministic operation which is true in our programming model.
\end{itemize}
\end{property}
\begin{property}\label{prop:wpsubset}
 Let $\phi_1$ and $\phi_2$ be the formulas such that $\phi_1$ logically implies $\phi_2$ then for every operation $op$, the formula $\weakestpre{op}{\phi_1}$ logically implies $\weakestpre{op}{\phi_2}$.
\end{property}

 We say that a formula $\phi$ is \emph{stable} with respect to a statement $S$ if $\weakestpre{S}{\phi}$ is logically equivalent to $\phi$.
 In this paper, we use weakest preconditions to check the correctness of a trace with respect to some safety assertion. A trace $\sigma$ reaching up to a safety assertion $\phi$ is safe if the execution of $\sigma$ starting from the initial state $\mathcal{I}$ either 1) blocks (does not terminate) because of not satisfying some path conditions, or 2) terminates and the resulting state satisfies $\phi$. The following lemmas clearly define the conditions, using weakest precondition axioms, for declaring a trace $\sigma$ either safe or unsafe.
Detailed proofs of these are given in Appendix \ref{prf:lemassumeassrtrhs2} and in \ref{prf:lemassumeassrtlhs2}. 
Here $\mathtt{\subst{\run}{assume}{assert}}$ denote the trace obtained by replacing every instruction of the form $\assume{\phi}$ by $\assrt{\phi}$ in $\run$.
 \begin{lemma}\label{lem:assumeassrtrhs2} 
 For a trace $\run$, an initial program state $\mathcal{I}$ and a safety property $\phi$, if $\weakestpre{\subst{\run}{assume}{assert}}{\neg \phi} \lland \mathcal{I}$ is unsatisfiable then the execution of $\run$, starting from $\mathcal{I}$, either does not terminate or terminates in a state satisfying $\phi$.
\end{lemma}
\begin{lemma}\label{lem:assumeassrtlhs2} 
 For a trace $\run$, an initial program state $\mathcal{I}$ and a safety property $\phi$, if $\weakestpre{\subst{\run}{assume}{assert}}{\neg \phi} \lland \mathcal{I}$ is satisfiable then the execution of $\run$, starting from $\mathcal{I}$, terminates in a state not satisfying $\phi$.
\end{lemma}
\subsection{Alternating Finite Automata (AFA)}\label{subsec:afa}
 Alternating finite automata \cite{DBLP:journals/tcs/BrzozowskiL80,Chandra:1981:ALT:322234.322243} are a generalization of nondeterministic finite automata (NFA).
An NFA is a five tuple $\langle S,\Sigma\cup \set{\epsilon},\delta,s_0,S_F\rangle$ with a set of states $S$,
ranged over by $s$, an initial state $s_0$, a set of accepting states $S_F$ and
a transition function $\delta: S \times \Sigma\cup\set{\epsilon} \to \Powerset{S}$. 
For any state $s$ of this NFA, the set of words accepted by $s$ is inductively defined as $\acc(s)=\set{a.\run \mid a\in \Sigma\cup \set{\epsilon}, \sigma \in \Sigma^*, \exists s'.~ s'\in \delta(s,a).~\run \in \acc(s')}$ where $\epsilon\in \acc(s)$ for all $s\in S_F$.
Here, the existential quantifier represents the fact that there should exist
at least one outgoing transition from $s$ along which $a.\run$ gets accepted.
An AFA is a six tuple $\langle \faS,\teS, \Sigma\cup \set{\epsilon},\delta,s_0,S_F \rangle$ with $\Sigma$, $s_0$ and $S_F\subseteq S$ denoting the alphabet, initial state and the set of accepting states respectively. $S=\faS \bigcup \teS$ is the set of all states, ranged over by $s$ and $\delta: S \times \Sigma\cup\set{\epsilon} \to \Powerset{S}$ is the transition function. The set of words accepted by a state of an AFA depends on whether that state is an \emph{existential state} (from the set $\teS$) or a \emph{universal state} (from the set $\faS$).
For an existential state $s\in \teS$, the set of accepted words is inductively defined in the same way as in NFA. For a universal state $s \in \faS$ the set of accepted words are $\acc(s)=\set{a.\run \mid a\in \Sigma\cup \set{\epsilon}, \forall s'\in \delta(s,a).~ \run \in \acc(s')}$ with $\epsilon \in \acc(s)$ for all $s\in S_F$. Notice the change in the quantifier from $\exists$ to $\forall$. In the diagrams of AFA used in this paper, we annotate universal states with $\forall$ symbol and existential states with $\exists$ symbol.
For a state $s$, let $\succset(s,a)=\set{S\mid (s,a,S)\in \delta}$ be the set of $a$-successors of $s$.
For an automaton $\mathcal{A}$, let $\lang{\mathcal{A}}$ be the language accepted by the initial state of that automaton.
For any $\run\in \Sigma^*$ $\len{\run}$ denote the length of $\run$ and $\rev(\sigma)$ denote the reverse of $\sigma$.

\newcommand{\stepone}{
\begin{tikzpicture}
\node[existsnode, label=below:$s_1$ ] (s1) {$\small{\neg (\ell_2=2)}$};
\fexstate{s1};
\setinitnode{s1};
\end{tikzpicture}
}

\newcommand{\steptwo}{
\begin{tikzpicture}
\node[existsnode, label=below:$s_1$ ] (s1) {$\small{\neg (\ell_2=2)}$};
\node[existsnode, right= of s1, label=below:$s_2$ ] (s2) {$\small{\neg (cs=2)}$};
\fexstate{s1};
\fexstate{s2};
\setinitnode{s1};
\draw [->] (s1) to node [above] {$s$} (s2) ;
\end{tikzpicture}
}

\newcommand{\stepthree}{
\begin{tikzpicture}
\node[existsnode, label=below:$s_1$ ] (s1) {$\small{\neg (\ell_2=2)}$};
\node[existsnode, right= of s1, label=below:$s_2$ ] (s2) {$\small{\neg (cs=2)}$};
\node[existsnode, right= of s2, label=below:$s_3$ ] (s3) {$\small{\true}$};
\fexstate{s1};
\fexstate{s2};
\fexstate{s3};
\setinitnode{s1};
\draw [->] (s1) to node [above] {$s$} (s2) ;
\draw [->] (s2) to node [above] {$c$} (s3) ;
\end{tikzpicture}
}

\newcommand{\stepfour}{
\begin{tikzpicture}
\node[existsnode, label=below:$s_1$ ] (s1) {$\small{\neg (\ell_2=2)}$};
\node[existsnode, right= of s1, label=below:$s_2$ ] (s2) {$\small{\neg (cs=2)}$};
\node[existsnode, right= of s2, label=below:$s_3$ ] (s3) {$\small{\true}$};
\fexstate{s1};
\fexstate{s2};
\fexstate{s3};
\draw [->] (s1) to node [above] {$s$} (s2) ;
\draw [->] (s2) to node [above] {$c$} (s3) ;
\draw [->, loop right=-1pt] (s3) to node [above] {$r$} (s3) ;
\end{tikzpicture}
}

\newcommand{\stepfive}{
\begin{tikzpicture}
\node[existsnode, label=below:$s_1$ ] (s1) {$\small{\neg (\ell_2=2)}$};
\node[existsnode, right= of s1, label=below left:$s_2$ ] (s2) {$\small{\neg (cs=2)}$};
\node[existsnode, below= of s2, label=below:$s_3$ ] (s3) {$\small{\true}$};
\node[existsnode, right= of s2, label=below right:$s_4$ ] (s4) {$\begin{array}{l}flag_1=0\\\llor turn=2\end{array}$};
\node[existsnode, below = of s4, label=left:$s_5$ ] (s5) {$flag_1=0$};
\node[existsnode, above = of s4, label=left:$s_6$ ] (s6) {$turn=2$};
\setinitnode{s1};
\fexstate{s1};
\fexstate{s2};
\fexstate{s3};
\fallstate{s4};
\fexstate{s5};
\fexstate{s6};
\draw [->] (s1) to node [above] {$s$} (s2) ;
\draw [->] (s2) to node [left] {$c$} (s3) ;
\draw [->] (s3) to node [above] {$P$} (s4) ;
\draw [->, loop right=-1pt] (s3) to node [above] {$r$} (s3) ;
\draw [->] (s3) to node [above] {$\epsilon$} (s4) ;
\draw [->] (s4) to node [left] {$\epsilon$} (s5) ;
\draw [->] (s4) to node [left] {$\epsilon$} (s6) ;
\end{tikzpicture}
}

\newcommand{\stepsix}{
\begin{tikzpicture}
\node[existsnode, label=below:$s_1$ ] (s1) {$\small{\neg (\ell_2=2)}$};
\node[existsnode, right= of s1, label=below:$s_2$ ] (s2) {$\small{\neg (cs=2)}$};
\node[existsnode, right= of s2, label=below:$s_3$ ] (s3) {$\begin{array}{l}flag_1=0\\\llor turn=2\end{array}$};
\node[existsnode, below right= of s3, label=below:$s_4$ ] (s4) {$flag_1=0$};
\node[existsnode, above right = of s3, label=below:$s_5$ ] (s5) {$turn=2$};
\fexstate{s1};
\fexstate{s2};
\fallstate{s3};
\fexstate{s4};
\fexstate{s5};
\setinitnode{s1};
\draw [->] (s1) to node [above] {$s$} (s2) ;
\draw [->] (s2) to node [above] {$c$} (s3) ;
\draw [->, loop right=-1pt] (s3) to node [above] {$r$} (s3) ;
\draw [->] (s3) to node [above] {$P$} (s4) ;
\draw [->] (s3) to node [above] {$P$} (s5) ;
\draw [->, loop right=-1pt] (s4) to node [above] {$p,q$} (s4) ;
\end{tikzpicture}
}

\newcommand{\stepseven}{
\begin{tikzpicture}
\node[existsnode, label=below:$s_1$ ] (s1) {$\small{\neg (\ell_2=2)}$};
\node[existsnode, right= of s1, label=below left:$s_2$ ] (s2) {$\small{\neg (cs=2)}$};
\node[existsnode, below= of s2, label=below:$s_3$ ] (s3) {$\small{\true}$};
\node[existsnode, right= of s2, label=below right:$s_4$ ] (s4) {$\begin{array}{l}flag_1=0\\\llor turn=2\end{array}$};
\node[existsnode, below = of s4, label=below:$s_5$ ] (s5) {$flag_1=0$};
\node[existsnode, above = of s4, label=below left:$s_6$ ] (s6) {$turn=2$};
\node[finalexistsnode, left= of s6, label=below:$s_7$ ] (s7) {$\small{\false}$};
\setinitnode{s1};
\fexstate{s1};
\fexstate{s2};
\fexstate{s3};
\fallstate{s4};
\fexstate{s5};
\fexstate{s6};
\draw [->] (s1) to node [above] {$s$} (s2) ;
\draw [->] (s2) to node [left] {$c$} (s3) ;
\draw [->] (s3) to node [above] {$P$} (s4) ;
\draw [->, loop right=-1pt] (s3) to node [above] {$r$} (s3) ;
\draw [->] (s3) to node [above] {$\epsilon$} (s4) ;
\draw [->] (s4) to node [left] {$\epsilon$} (s5) ;
\draw [->] (s4) to node [left] {$\epsilon$} (s6) ;
\draw [->] (s6) to node [above] {$q$} (s7) ;
\draw [->, loop left=-1pt] (s7) to node [left] {$p,a,b,A$} (s7) ;
\draw [->, loop right=-1pt] (s5) to node [right] {$p,q$} (s5) ;

\end{tikzpicture}
}

\newcommand{\stepeight}{
\begin{tikzpicture}
\node[existsnode, label=below:$s_1$ ] (s1) {$\small{\neg (\ell_2=2)}$};
\node[existsnode, right= of s1, label=below:$s_2$ ] (s2) {$\small{\neg (cs=2)}$};
\node[existsnode, right= of s2, label=below:$s_3$ ] (s3) {$\begin{array}{l}flag_1=0\\\llor turn=2\end{array}$};
\node[existsnode, below right = of s3, label=below:$s_4$ ] (s4) {$\begin{array}{l}flag_1=0 \land\\ (flag_2=0\\\llor turn=1)\end{array}$};
\node[existsnode, above right = of s3, label=below:$s_5$ ] (s5) {$turn=2$};
\node[finalexistsnode, above = of s5, label=below:$s_6$ ] (s6) {$\false$};
\node[existsnode, below right = of s4, label=below:$s_7$ ] (s7) {$flag_1=0$};
\node[existsnode, right =of s4, label=below:$s_8$ ] (s8) {$\begin{array}{l}flag_2=0\\ \llor turn=1\end{array}$};
\node[finalexistsnode, below right=of s8, label=below:$s_9$ ] (s9) {$flag_2=0$};
\node[existsnode, above right=of s8, label=below:$s_{10}$ ] (s10) {$turn=1$};
\fexstate{s1};
\fexstate{s2};
\fallstate{s3};
\fallstate{s4};
\fexstate{s5};
\fexstate{s6};
\fexstate{s7};
\fallstate{s8};
\fexstate{s9};
\fexstate{s10};

\setinitnode{s1};
\draw [->] (s1) to node [above] {$s$} (s2) ;
\draw [->] (s2) to node [above] {$c$} (s3) ;
\draw [->, loop right=-1pt] (s3) to node [above] {$r$} (s3) ;
\draw [->] (s3) to node [above] {$P$} (s4) ;
\draw [->] (s3) to node [above] {$P$} (s5) ;
\draw [->, loop right=-1pt] (s4) to node [above] {$p,q$} (s4) ;
\draw [->] (s5) to node [above] {$q$} (s6) ;
\draw [->] (s4) to node [above] {$A$} (s7) ;
\draw [->] (s4) to node [above] {$A$} (s8) ;
\draw [->] (s8) to node [above] {$\epsilon$} (s9) ;
\draw [->] (s8) to node [above] {$\epsilon$} (s10) ;
\draw [->, loop right=-1pt] (s9) to node [above] {$a,b$} (s9) ;
\draw [->, loop right=-1pt] (s6) to node [above] {$p,a,b,A$} (s6) ;
\end{tikzpicture}
}

\newcommand{\stepnine}{
\begin{tikzpicture}
\node[existsnode, label=below:$s_1$ ] (s1) {$\small{\neg (\ell_2=2)}$};
\node[existsnode, right= of s1, label=below:$s_2$ ] (s2) {$\small{\neg (cs=2)}$};
\node[existsnode, right= of s2, label=below:$s_3$ ] (s3) {$\begin{array}{l}flag_1=0\\\llor turn=2\end{array}$};
\node[existsnode, below right = of s3, label=below:$s_4$ ] (s4) {$\begin{array}{l}flag_1=0 \land\\ (flag_2=0\\\llor turn=1)\end{array}$};
\node[existsnode, above right = of s3, label=below:$s_5$ ] (s5) {$turn=2$};
\node[finalexistsnode, right = of s5, label=below:$s_6$ ] (s6) {$\false$};
\node[existsnode, below right = of s4, label=below:$s_7$ ] (s7) {$flag_1=0$};
\node[existsnode, right =of s4, label=below:$s_8$ ] (s8) {$\begin{array}{l}flag_2=0\\ \llor turn=1\end{array}$};
\node[finalexistsnode, below right=of s8, label=below:$s_9$ ] (s9) {$flag_2=0$};
\node[existsnode, above right=of s8, label=below:$s_{10}$ ] (s10) {$turn=1$};
\node[finalexistsnode, right =of s10, label=below:$s_{11}$ ] (s11) {$\false$};
\fexstate{s1};
\fexstate{s2};
\fallstate{s3};
\fallstate{s4};
\fexstate{s5};
\fexstate{s6};
\fexstate{s7};
\fallstate{s8};
\fexstate{s9};
\fexstate{s10};
\fexstate{s11};
\setinitnode{s1};
\draw [->] (s1) to node [above] {$s$} (s2) ;
\draw [->] (s2) to node [above] {$c$} (s3) ;
\draw [->, loop right=-1pt] (s3) to node [above] {$r$} (s3) ;
\draw [->] (s3) to node [above] {$P$} (s4) ;
\draw [->] (s3) to node [above] {$P$} (s5) ;
\draw [->, loop right=-1pt] (s4) to node [above] {$p,q$} (s4) ;
\draw [->] (s5) to node [above] {$q$} (s6) ;
\draw [->] (s4) to node [above] {$A$} (s7) ;
\draw [->] (s4) to node [above] {$A$} (s8) ;
\draw [->] (s8) to node [above] {$\epsilon$} (s9) ;
\draw [->] (s8) to node [above] {$\epsilon$} (s10) ;
\draw [->] (s10) to node [above] {$b$} (s11) ;
\draw [->, loop right=-1pt] (s9) to node [above] {$a,b$} (s9) ;
\draw [->, loop right=-1pt] (s11) to node [above] {$a$} (s11) ;
\draw [->, loop right=-1pt] (s6) to node [above] {$p,a,b,A$} (s6) ;
\end{tikzpicture}
}

\newcommand{\stepten}{
\begin{tikzpicture}
\node[existsnode, label=below:$s_1$ ] (s1) {$\small{\neg (\ell_2=2)}$};
\node[existsnode, right= of s1, label=below:$s_2$ ] (s2) {$\small{\neg (cs=2)}$};
\node[existsnode, right= of s2, label=below:$s_3$ ] (s3) {$\begin{array}{l}flag_1=0\\\llor turn=2\end{array}$};
\node[existsnode, below right = of s3, label=below:$s_4$ ] (s4) {$\begin{array}{l}flag_1=0 \land\\ (flag_2=0\\\llor turn=1)\end{array}$};
\node[existsnode, above right = of s3, label=below:$s_5$ ] (s5) {$turn=2$};
\node[finalexistsnode, right = of s5, label=below:$s_6$ ] (s6) {$\false$};
\node[existsnode, below right = of s4, label=below:$s_7$ ] (s7) {$flag_1=0$};
\node[existsnode, right =of s4, label=below:$s_8$ ] (s8) {$\begin{array}{l}flag_2=0\\ \llor turn=1\end{array}$};
\node[finalexistsnode, below right=of s8, label=below:$s_9$ ] (s9) {$flag_2=0$};
\node[existsnode, above right=of s8, label=below:$s_{10}$ ] (s10) {$turn=1$};
\node[finalexistsnode, right =of s10, label=below:$s_{11}$ ] (s11) {$\false$};
\node[finalexistsnode, below right =of s7, label=below:$s_{12}$ ] (s12) {$\false$};
\fexstate{s1};
\fexstate{s2};
\fallstate{s3};
\fallstate{s4};
\fexstate{s5};
\fexstate{s6};
\fexstate{s7};
\fallstate{s8};
\fexstate{s9};
\fexstate{s10};
\fexstate{s11};
\fexstate{s12};
\setinitnode{s1};
\draw [->] (s1) to node [above] {$s$} (s2) ;
\draw [->] (s2) to node [above] {$c$} (s3) ;
\draw [->, loop right=-1pt] (s3) to node [above] {$r$} (s3) ;
\draw [->] (s3) to node [above] {$P$} (s4) ;
\draw [->] (s3) to node [above] {$P$} (s5) ;
\draw [->, loop right=-1pt] (s4) to node [above] {$p,q$} (s4) ;
\draw [->] (s5) to node [above] {$q$} (s6) ;
\draw [->] (s4) to node [above] {$A$} (s7) ;
\draw [->] (s4) to node [above] {$A$} (s8) ;
\draw [->] (s8) to node [above] {$\epsilon$} (s9) ;
\draw [->] (s8) to node [above] {$\epsilon$} (s10) ;
\draw [->] (s10) to node [above] {$b$} (s11) ;
\draw [->] (s7) to node [above] {$a$} (s12) ;
\draw [->, loop right=-1pt] (s9) to node [above] {$a,b$} (s9) ;
\draw [->, loop left=-1pt] (s7) to node [above] {$b$} (s7) ;
\draw [->, loop right=-1pt] (s11) to node [above] {$a$} (s11) ;
\draw [->, loop right=-1pt] (s6) to node [above] {$p,a,b,A$} (s6) ;
\end{tikzpicture}
}


\newcommand{\stepeightvert}{
\begin{tikzpicture}
\node[existsnode, label=below:$s_1$ ] (s1) {$\small{\neg (\ell_2=2)}$};
\node[existsnode, right= of s1, label=below:$s_2$ ] (s2) {$\small{\neg (cs=2)}$};
\node[existsnode, right= of s2, label=below:$s_3$ ] (s3) {$\small{\true}$};
\node[existsnode, above= of s3, label=left:$s_4$] (s4) {$\begin{array}{l}flag_1=0\\\llor turn=2\end{array}$};
\node[existsnode, above left= of s4, label=left:$s_5$] (s5) {$flag_1=0$};
\node[existsnode, above = of s4, label=below right:$s_6$] (s6) {$turn=2$};
\node[finalexistsnode, above = of s6, label=left:$s_7$] (s7) {$\small{\false}$};
\node[existsnode, above = of s5, label=left:$s_8$] (s8) {$\begin{array}{l}flag_1=0 \land\\ (flag_2=0\\\llor turn=1)\end{array}$};
\node[existsnode, above left= of s8, label=left:$s_9$] (s9) {$flag_1=0$};
\node[existsnode, above right= of s9, label=left:$s_{10}$] (s10) {$\begin{array}{l}flag_2=0\\\llor turn=1\end{array}$};
\node[finalexistsnode, above left= of s10, label=left:$s_{11}$] (s11) {$flag_2=0$};
\node[existsnode, above right= of s10, label=left:$s_{12}$] (s12) {$turn=1$};
\fexstate{s1};
\fexstate{s2};
\fallstate{s3};
\fallstate{s4};
\fexstate{s5};
\fexstate{s6};
\fexstate{s7};
\fallstate{s8};
\fexstate{s9};
\fallstate{s10};
\fexstate{s11};
\fexstate{s12};
\setinitnode{s1};
\draw [->] (s1) to node [above] {$s$} (s2) ;
\draw [->] (s2) to node [above] {$c$} (s3) ;
\draw [->] (s3) to node [left] {$P$} (s4) ;
\draw [->, loop right=-1pt] (s3) to node [above] {$r$} (s3) ;
\draw [->] (s4) to node [right] {$\epsilon$} (s5) ;
\draw [->] (s4) to node [right] {$\epsilon$} (s6) ;
\draw [->, loop right=-1pt] (s5) to node [above	] {$p,q$} (s5) ;
\draw [->] (s5) to node [right] {$A$} (s8) ;
\draw [->] (s6) to node [right] {$q$} (s7);
\draw [->, loop right=-1pt] (s7) to node [right] {$p,a,b,A$} (s7) ;
\draw [->] (s8) to node [right] {$\epsilon$} (s9) ;
\draw [->] (s8) to node [right] {$\epsilon$} (s10) ;
\draw [->] (s10) to node [right] {$\epsilon$} (s11) ;
\draw [->] (s10) to node [right] {$\epsilon$} (s12) ;
\draw [->, loop right=-1pt] (s11) to node [above] {$a,b$} (s11) ;
\end{tikzpicture}
}

\newcommand{\stepninevert}{
\begin{tikzpicture}
\node[existsnode, label=below:$s_1$ ] (s1) {$\small{\neg (\ell_2=2)}$};
\node[existsnode, right= of s1, label=below:$s_2$ ] (s2) {$\small{\neg (cs=2)}$};
\node[existsnode, right= of s2, label=below:$s_3$ ] (s3) {$\small{\true}$};
\node[existsnode, above= of s3, label=left:$s_4$] (s4) {$\begin{array}{l}flag_1=0\\\llor turn=2\end{array}$};
\node[existsnode, above left= of s4, label=left:$s_5$] (s5) {$flag_1=0$};
\node[existsnode, above = of s4, label=below right:$s_6$] (s6) {$turn=2$};
\node[existsnode, above = of s5, label=left:$s_8$] (s8) {$\begin{array}{l}flag_1=0 \land\\ (flag_2=0\\\llor turn=1)\end{array}$};
\node[existsnode, above left= of s8, label=left:$s_9$] (s9) {$flag_1=0$};
\node[existsnode, above right= of s9, label=left:$s_{10}$] (s10) {$\begin{array}{l}flag_2=0\\\llor turn=1\end{array}$};
\node[finalexistsnode, above left= of s10, label=left:$s_{11}$] (s11) {$flag_2=0$};
\node[existsnode, above right= of s10, label=left:$s_{12}$] (s12) {$turn=1$};
\node[finalexistsnode, above= of s9, label=left:$s_{13}$] (s13) {$\small{\false}$};
\node[finalexistsnode, below= of s12, label=below:$s_{14}$] (s14) {$\small{\false}$};
\node[finalexistsnode, above= of s6, label=left:$s_{7}$] (s7) {$\small{\false}$};
\fexstate{s1};
\fexstate{s2};
\fallstate{s3};
\fallstate{s4};
\fexstate{s5};
\fexstate{s6};
\fexstate{s7};
\fallstate{s8};
\fexstate{s9};
\fallstate{s10};
\fexstate{s11};
\fexstate{s12};
\fexstate{s13};
\fexstate{s14};
\setinitnode{s1};
\draw [->] (s1) to node [above] {$s$} (s2) ;
\draw [->] (s2) to node [above] {$c$} (s3) ;
\draw [->] (s3) to node [left] {$P$} (s4) ;
\draw [->, loop right=-1pt] (s3) to node [above] {$r$} (s3) ;
\draw [->] (s4) to node [right] {$\epsilon$} (s5) ;
\draw [->] (s4) to node [right] {$\epsilon$} (s6) ;
\draw [->, loop right=-1pt] (s5) to node [above	] {$p,q$} (s5) ;
\draw [->] (s5) to node [right] {$A$} (s8) ;
\draw [->] (s8) to node [right] {$\epsilon$} (s9) ;
\draw [->] (s8) to node [right] {$\epsilon$} (s10) ;
\draw [->] (s10) to node [right] {$\epsilon$} (s11) ;
\draw [->] (s10) to node [right] {$\epsilon$} (s12) ;
\draw [->] (s9) to node [right] {$a$} (s13) ;
\draw [->] (s12) to node [right] {$b$} (s14) ;
\draw [->] (s6) to node [right] {$q$} (s7) ;
\draw [->, loop right=-1pt] (s7) to node [right] {$p,a,b,A$} (s7) ;
\draw [->, loop right=-1pt] (s9) to node [above] {$b$} (s9) ;
\draw [->, loop right=-1pt] (s14) to node [above] {$a$} (s14) ;
\draw [->, loop right=-1pt] (s11) to node [above] {$a,b$} (s11) ;
\end{tikzpicture}
}

\newcommand{\stepninehorizon}[1]{
\begin{tikzpicture}
\node[existsnode, label=below right:$s_0$, label={[ellipse ,draw]above:{$\false$}} ] (s1) {$\small{\neg (\ell_2=2)}$};
\node[existsnode, below= of s1, label=below right :$s_1$, label={[ellipse ,draw]left:{$\false$}}  ] (s2) {$\small{\neg \mathrm{(res=2)}}$};
\node[existsnode, below= of s2, label=below right:$s_2$, label={[ellipse ,draw]left:{$\false$}}  ] (s3) {$\small{\true}$};
\node[existsnode, right= of s3, label=below:$s_3$, label={[ellipse ,draw]above:{$\false$}} ] (s4) {$\begin{array}{l}\mathrm{flag_1=\false}\\\llor \mathrm{turn=2}\end{array}$};
\node[existsnode, above right = of s4, label=left:$s_4$, label={[ellipse ,draw]below right:{$\false$}} ] (s5) {$\mathrm{flag_1=\false}$};
\node[existsnode, below right= of s4, label=left:$s_5$, label={[ellipse ,draw]above:{$\false$}} ] (s6) {$\mathrm{turn=2}$};
\node[existsnode, above = of s5, label=below left:$s_7$, label={[ellipse ,draw]left:{$\false$}} ] (s8) {$\begin{array}{l}\mathrm{flag_1=\false} \land\\ (\mathrm{flag_2=\false}\\\llor \mathrm{turn=1})\end{array}$};
\node[existsnode, above = of s8, label=below left:$s_8$, label={[ellipse ,draw]right:{$\false$}} ] (s9) {$\mathrm{flag_1=\false}$};
\node[existsnode, below right= of s9, label=below right:$s_{9}$, label={[ellipse ,draw]right:{$\mathrm{flag_2=\false}$}} ] (s10) {$\begin{array}{l}\mathrm{flag_2=\false}\\\llor \mathrm{turn=1}\end{array}$};
\node[finalexistsnode, above = of s10, label=below right:$s_{10}$, label={[ellipse ,draw]right:{$\mathrm{flag_2=\false}$}} ] (s11) {$\mathrm{flag_2=\false}$};
\node[existsnode, below = of s10, label=below right:$s_{11}$, label={[ellipse ,draw]above left:{$\false$}} ] (s12) {$\mathrm{turn=1}$};
\node[finalexistsnode, left= of s9, label=below:$s_{12}$, label={[ellipse ,draw]above:{$\false$}} ] (s13) {$\small{\false}$};
\node[finalexistsnode, below= of s12, label=below:$s_{13}$, label={[ellipse ,draw]right:{$\false$}} ] (s14) {$\small{\false}$};
\node[finalexistsnode, right= of s6, label=below:$s_{6}$, label={[ellipse ,draw]right:{$\false$}} ] (s7) {$\small{\false}$};
\fexstate{s1};
\fexstate{s2};
\fexstate{s3};
\fallstate{s4};
\fexstate{s5};
\fexstate{s6};
\fexstate{s7};
\fallstate{s8};
\fexstate{s9};
\fallstate{s10};
\fexstate{s11};
\fexstate{s12};
\fexstate{s13};
\fexstate{s14};
\setinitnode{s1};
\draw [->] (s1) to node [left] {$\lb{s}$} (s2) ;
\draw [->, loop right=-3pt] (s1) to node [right	] {$\Sigma\setminus\set{\lb{s,A,P}}$} (s1) ;
\draw [->] (s2) to node [left] {$\lb{c}$} (s3) ;
\draw [->, loop right=-3pt] (s2) to node [right	] {$\Sigma\setminus\set{\lb{c,r,A,P}}$} (s2) ;
\draw [->] (s3) to node [above] {$\lb{P}$} (s4) ;
\draw [->, loop below=-3pt] (s3) to node [left] {$\Sigma\setminus\set{\lb{A,P}}$} (s3) ;
\draw [->] (s4) to node [left] {$\epsilon$} (s5) ;
\draw [->] (s4) to node [right] {$\epsilon$} (s6) ;
\draw [->, loop below=-1pt] (s5) to node [below] {$\Sigma\setminus\set{\lb{a,e,A,P}}$} (s5) ;
\draw [->, loop below=-1pt] (s6) to node [below] {$\Sigma\setminus\set{\lb{b,q,A,P}}$} (s6) ;
\draw [->, loop right=-1pt] (s12) to node [above] {$\Sigma\setminus\set{\lb{b,q,A,P}}$} (s12) ;
\draw [->] (s5) to node [right] {$\lb{A}$} (s8) ;
\draw [->] (s8) to node [right] {$\epsilon$} (s9) ;
\draw [->] (s8) to node [above] {$\epsilon$} (s10) ;
\draw [->] (s10) to node [right] {$\epsilon$} (s11) ;
\draw [->] (s10) to node [right] {$\epsilon$} (s12) ;
\draw [->] (s9) to node [above] {$\lb{a}$} (s13) ;
\draw [->] (s12) to node [right] {$\lb{b}$} (s14) ;
\draw [->] (s6) to node [above] {$\lb{q}$} (s7) ;
\draw [->, loop above=-1pt] (s7) to node [above] {$\Sigma$} (s7) ;
\draw [->, loop left=-1pt] (s13) to node [above] {$\Sigma$} (s13) ;
\draw [->, loop above=-1pt] (s9) to node [above	] {$\Sigma\setminus\set{\lb{a,e,A,P}}$} (s9) ;
\draw [->, loop left=-1pt] (s14) to node [left] {$\Sigma$} (s14) ;
\draw [->, loop above=-1pt] (s11) to node [above] {$\Sigma\setminus\set{\lb{p,t,A,P}}$} (s11) ;
\end{tikzpicture}
}

\newcommand{\stepninehorizononly}[1]{
\begin{tikzpicture}
\node[existsnode, label=below right:$s_0$, label={[ellipse ,draw]above:{$\lb{abApqPrcs}$}} ] (s1) {$\small{\neg (\ell_2=2)}$};
\node[existsnode, below= of s1, label=below right :$s_1$, label={[ellipse ,draw]left:{$\lb{abApqPrc}$}}  ] (s2) {$\small{\neg (\mathrm{res=2})}$};
\node[existsnode, below= of s2, label=below right:$s_2$, label={[ellipse ,draw]left:{$\lb{abApqPr}$}}  ] (s3) {$\small{\true}$};
\node[existsnode, right= of s3, label=below left:$s_3$, label={[ellipse ,draw]above:{$\lb{abApq}$}} ] (s4) {$\begin{array}{l}\mathrm{flag_1=\false}\\\llor \mathrm{turn=2}\end{array}$};
\node[existsnode, above right = of s4, label=above left:$s_4$, label={[ellipse ,draw]below:{$\lb{abApq}$}} ] (s5) {$\mathrm{flag_1=\false}$};
\node[existsnode, below right= of s4, label=left:$s_5$, label={[ellipse ,draw]above:{$\lb{abApq}$}} ] (s6) {$\mathrm{turn=2}$};
\node[existsnode, right = of s5, label=below left:$s_7$, label={[ellipse ,draw]below:{$\lb{ab}$}} ] (s8) {$\begin{array}{l}\mathrm{flag_1=\false} \land\\ (\mathrm{flag_2=\false}\\\llor \mathrm{turn=1})\end{array}$};
\node[existsnode, above = of s8, label=below left:$s_8$, label={[ellipse ,draw]right:{$\lb{ab}$}} ] (s9) {$\mathrm{flag_1=\false}$};
\node[existsnode, below right= of s9, label=below right:$s_{9}$, label={[ellipse ,draw]right:{$\lb{ab}$}} ] (s10) {$\begin{array}{l}\mathrm{flag_2=\false}\\\llor \mathrm{turn=1}\end{array}$};
\node[finalexistsnode, above = of s10, label=below right:$s_{10}$, label={[ellipse ,draw]right:{$\lb{ab}$}} ] (s11) {$\mathrm{flag_2=\false}$};
\node[existsnode, below = of s10, label=below right:$s_{11}$, label={[ellipse ,draw]right:{$\lb{ab}$}} ] (s12) {$\mathrm{turn=1}$};
\node[finalexistsnode, left= of s9, label=below:$s_{12}$ ] (s13) {$\small{\false}$};
\node[finalexistsnode, below= of s12, label=below:$s_{13}$, label={[ellipse ,draw]right:{$\lb{a}$}} ] (s14) {$\small{\false}$};
\node[finalexistsnode, right= of s6, label=below:$s_{6}$, label={[ellipse ,draw]right:{$\lb{abAp}$}} ] (s7) {$\small{\false}$};
\fexstate{s1};
\fexstate{s2};
\fexstate{s3};
\fallstate{s4};
\fexstate{s5};
\fexstate{s6};
\fexstate{s7};
\fallstate{s8};
\fexstate{s9};
\fallstate{s10};
\fexstate{s11};
\fexstate{s12};
\fexstate{s13};
\fexstate{s14};
\setinitnode{s1};
\draw [->] (s1) to node [left] {$\lb{s}$} (s2) ;
\draw [->, loop right=-3pt] (s1) to node [right	] {$\Sigma\setminus\set{\lb{s,A,P}}$} (s1) ;
\draw [->] (s2) to node [left] {$c$} (s3) ;
\draw [->, loop right=-3pt] (s2) to node [right	] {$\Sigma\setminus\set{\lb{c,r,A,P}}$} (s2) ;
\draw [->] (s3) to node [above] {$P$} (s4) ;
\draw [->, loop below=-3pt] (s3) to node [left] {$\Sigma\setminus\set{\lb{A,P}}$} (s3) ;
\draw [->] (s4) to node [right] {$\epsilon$} (s5) ;
\draw [->] (s4) to node [right] {$\epsilon$} (s6) ;
\draw [->, loop above=-1pt] (s5) to node [above	] {$\Sigma\setminus\set{\lb{a,e,A,P}}$} (s5) ;
\draw [->, loop below=-1pt] (s6) to node [below] {$\Sigma\setminus\set{\lb{b,q,A,P}}$} (s6) ;
\draw [->, loop left=-1pt] (s12) to node [below] {$\Sigma\setminus\set{\lb{b,q,A,P}}$} (s12) ;
\draw [->] (s5) to node [above] {$\lb{A}$} (s8) ;
\draw [->] (s8) to node [right] {$\epsilon$} (s9) ;
\draw [->] (s8) to node [above] {$\epsilon$} (s10) ;
\draw [->] (s10) to node [right] {$\epsilon$} (s11) ;
\draw [->] (s10) to node [right] {$\epsilon$} (s12) ;
\draw [->] (s9) to node [above] {$\lb{a}$} (s13) ;
\draw [->] (s12) to node [right] {$\lb{b}$} (s14) ;
\draw [->] (s6) to node [above] {$\lb{q}$} (s7) ;
\draw [->, loop above=-1pt] (s7) to node [above] {$\Sigma$} (s7) ;
\draw [->, loop left=-1pt] (s13) to node [above] {$\Sigma$} (s13) ;
\draw [->, loop above=-1pt] (s9) to node [above	] {$\Sigma\setminus\set{\lb{a,e,A,P}}$} (s9) ;
\draw [->, loop left=-1pt] (s14) to node [left] {$\Sigma$} (s14) ;
\draw [->, loop above=-1pt] (s11) to node [above] {$\Sigma\setminus\set{\lb{p,t,A,P}}$} (s11) ;
\end{tikzpicture}
}

\newcommand{\stepninehorizonhmapextra}{
\begin{tikzpicture}
\node[existsnode, label=below right:$s_0$, label={[ellipse ,draw]above:{$\false$}} ] (s1) {$\small{\neg (\ell_2=2)}$};
\node[existsnode, below= of s1, label=below right :$s_1$, label={[ellipse ,draw]left:{$\false$}}  ] (s2) {$\small{\neg \mathrm{(res=2)}}$};
\node[existsnode, below= of s2, label=below right:$s_2$, label={[ellipse ,draw]left:{$\false$}}  ] (s3) {$\small{\true}$};
\node[existsnode, right= of s3, label=below:$s_3$, label={[ellipse ,draw]above:{$\false$}} ] (s4) {$\begin{array}{l}\mathrm{flag_1=\false}\\\llor \mathrm{turn=2}\end{array}$};
\node[existsnode, above right = of s4, label=left:$s_4$, label={[ellipse ,draw]below:{$\false$}} ] (s5) {$\mathrm{flag_1=\false}$};
\node[existsnode, below right= of s4, label=below right:$s_5$, label={[ellipse ,draw]above:{$\false$}} ] (s6) {$\mathrm{turn=2}$};
\node[existsnode, below right = of s5, label=below left:$s_7$, label={[ellipse ,draw]below:{$\false$}} ] (s8) {$\begin{array}{l}\mathrm{flag_1=\false} \land\\ (\mathrm{flag_2=\false}\\\llor \mathrm{turn=1})\end{array}$};
\node[existsnode, above  = 3 cm of s8, label=below right:$s_8$, label={[ellipse ,draw]right:{$\false$}} ] (s9) {$\mathrm{flag_1=\false}$};
\node[finalexistsnode, left= of s9, label=below:$s_{12}$, label={[ellipse ,draw]above:{$\false$}} ] (s13) {$\small{\false}$};
%
\node[finalexistsnode, above = of s8, label=below right:$s_{u}$, label={[ellipse ,draw]right:{$\false$}} ] (su) {$\small{\mathrm{flag_1=\false}}$};
\node[finalexistsnode, left= of s6, label=below:$s_{6}$, label={[ellipse ,draw]left:{$\false$}} ] (s7) {$\small{\false}$};
\fexstate{s1};
\fexstate{s2};
\fexstate{s3};
\fallstate{s4};
\fallstate{su};
\fexstate{s5};
\fexstate{s6};
\fexstate{s7};
\fexstate{s8};
\fexstate{s9};
\fexstate{s13};
\setinitnode{s1};
\draw [->] (s1) to node [left] {$\lb{s}$} (s2) ;
\draw [->, loop right=-3pt] (s1) to node [right	] {$\Sigma\setminus\set{\lb{s,A,P}}$} (s1) ;
\draw [->] (s2) to node [left] {$\lb{c}$} (s3) ;
\draw [->, loop right=-3pt] (s2) to node [right	] {$\Sigma\setminus\set{\lb{c,r,A,P}}$} (s2) ;
\draw [->] (s3) to node [above] {$P$} (s4) ;
\draw [->, loop below=-3pt] (s3) to node [left] {$\Sigma\setminus\set{\lb{A,P}},\lb{A}$} (s3) ;
\draw [->, bend left=-5] (s5) to node [above] {$\epsilon$} (s9) ;
\draw [->] (s4) to node [right] {$\epsilon$} (s5) ;
\draw [->] (s4) to node [right] {$\epsilon$} (s6) ;
\draw [->, loop above=-1pt] (s5) to node [above	] {$\Sigma\setminus\set{\lb{a,e,A,P}}$} (s5) ;
\draw [->, loop below=-1pt] (s6) to node [below] {$\Sigma\setminus\set{\lb{b,q,A,P}}$} (s6) ;
\draw [->] (s5) to node [above] {$\lb{A}$} (s8) ;
\draw [->] (s8) to node [right] {$\epsilon$} (su) ;
\draw [->] (su) to node [right] {$\epsilon$} (s9) ;
\draw [->] (s9) to node [above] {$\lb{a}$} (s13) ;
\draw [->] (s6) to node [above] {$\lb{q}$} (s7) ;
\draw [->, loop above=-1pt] (s7) to node [above] {$\Sigma$} (s7) ;
\draw [->, loop left=-1pt] (s13) to node [above] {$\Sigma$} (s13) ;
\draw [->, loop above=-1pt] (s9) to node [above	] {$\Sigma\setminus\set{\lb{a,e,A,P}}, \lb{P}$} (s9) ;
\end{tikzpicture}
}


\section{Our Approach}\label{sec:example}
The overall approach of this paper can be described in the following steps:
(i) Given a concurrent program $\mathcal{P}$, construct all its interleaved traces represented by automaton $\aut{\mathcal{P}}$, as defined in Subsection \ref{subsec:pro};
(ii)   Pick a trace $\sigma$ and a safety property, say $\phi$, to prove for this trace;
(iii) Prove $\sigma$ correct with respect to $\phi$ using Lemma \ref{lem:assumeassrtrhs2} and Lemma \ref{lem:assumeassrtlhs2} and generate a set of traces which are also \emph{provably correct}. Let us call this set $Tr'$;
(iv) Remove set $Tr'$ from the set of traces represented by $\aut{\mathcal{P}}$ and repeat from Step (ii) until either all the traces in $\mathcal{P}$ are proved correct or an erroneous trace is found.

Step (iii) of this procedure, correctness of $\sigma$, can be achieved by checking the unsatisfiability of $\weakestpre{\subst{\sigma}{assume}{assert}}{\neg \phi} \land \mathcal{I}$. However, we are not only interested in checking the correctness of $\sigma$ but also in constructing a set of traces which have a similar reasoning as of $\sigma$. Therefore, instead of computing $\weakestpre{\subst{\sigma}{assume}{assert}}{\neg \phi}$ directly from the weakest precondition axioms of Figure \ref{fig:wpstmts}, we construct an AFA from $\sigma$ and $\neg \phi$. Step (iv) is then achieved by applying automata-theoretic operations such as complementation and subtraction on this AFA. Notion of universal and existential states of AFA helps us in finding a set of sufficient dependencies used in the weakest precondition computation so that any other trace satisfying those dependencies gets captured by AFA. Subsequent subsections covers the construction, properties and use of this AFA in detail.
\newcommand{\NextState}{\mathrm{Next}}
\subsection{Constructing the AFA from a Trace and a Formula}
\begin{definition}\label{def:afa}
An AFA constructed from a trace $\sigma$ of a Program $P$ and a formula $\phi$ is $\afa{\run,\phi}=\langle \faS,\teS,\mathcal{OP_{\epsilon}},s_0,S_F,\delta,\AssnMap,\ResdMap\rangle$, where,
\begin{enumerate}
\item $(\mathcal{OP}_{\epsilon}=\mathcal{OP} \cup \{\epsilon\})$ is the alphabet ranged over by $op$. Here $\mathcal{OP}$ is the set of instructions used in program $P$. Symbol $\epsilon$ acts as an identity element of concatenation and $\weakestpre{\epsilon}{\phi}=\phi$. 
 \item $S=\faS \bigcup \teS$ is the largest set of states, ranged over by $s$ s.t.
  \begin{enumerate}
 \item Every state is annotated with a formula and a prefix of $\sigma$ denoted by $\AssnMap(s)$ and $\ResdMap(s)$ respectively. State $s_0$ is the initial state such that $\AssnMap(s_0)=\phi$, $\ResdMap(s_0)=\sigma$.
  \item \label{create}$s'\in S$ iff either of the following two conditions hold,
 \begin{itemize}
 \item $\exists s\in S$ such that $\AssnMap(s')$ is $\weakestpre{\subst{op}{assume}{assert}}{\AssnMap(s)}$, $\ResdMap(s)=\ResdMap(s').op.\sigma'$ and $\sigma'$ is the largest suffix of $\ResdMap(s)$ such that formula $\AssnMap(s)$ is stable with respect to $\subst{\sigma'}{assume}{assert}$.
 \item $\exists s\in S$ such that $\AssnMap(s)=\bigwedge\set{\phi_1,\cdots,\phi_k}$ or $\AssnMap(s)=\bigvee \set{\phi_1,\cdots,\phi_k}$, $\ResdMap(s)=\ResdMap(s')$, $\AssnMap(s')=\phi'$ and $\phi' \in \set{\phi_1,\cdots,\phi_k}$.
 \end{itemize}
 \item \label{uniexist} A state $s\in S$ is an existential state (universal state) iff $\AssnMap(s)$ is a literal (compound formula). 
 \end{enumerate}
 \item\label{fin} $S_F\subseteq S$ is a set of accepting states such that $s\in S_F$ $\textrm{iff}$ $\weakestpre{\subst{\ResdMap(s)}{assume}{assert}}{\AssnMap(s)}$ is same as $\AssnMap(s)$, i.e. $\AssnMap(s)$ is \emph{stable} with respect to $\subst{\ResdMap(s)}{assume}{assert}$, and
 \item Function $\delta:\!S\!\times\!\mathcal{OP}_{\epsilon}\!\to\!\Powerset{S}$ is defined in Figure \ref{fig:afa:nextst}.
 \end{enumerate}
%
%
 \end{definition}
 \begin{figure}
\scalebox{0.75}{\parbox{\linewidth}{%
\begin{subequations}  
$~~~\delta(s,op)=$
\begin{align*}
       \tikzmark{fsttt}{} & \{s'\} && \mbox{ if }\left\{\begin{array}{l}1.\AssnMap(s')=\weakestpre{\subst{op}{assume}{assert}}{\AssnMap(s)}, \\
       2. s \mbox{ is an existential state, and}\\
       3.\ResdMap(s)=\ResdMap(s').op.\sigma'' \\\mbox{\textbf{where} $\sigma''$ is the longest sequence s.t.} \\ \weakestpre{\subst{\sigma''}{assume}{assert}}{\AssnMap(s)}=\AssnMap(s) \end{array}\right.\tag{\textsc{Literal-Assn}}\label{tag:lit-assn}\\
       & \{s\} && \mbox{ if }\left\{\begin{array}{l}1.\AssnMap(s)=\weakestpre{\subst{op}{assume}{assert}}{\AssnMap(s)}, \mbox{ and}\\
	2.s \mbox{ is an existential state}\end{array}\right.\tag{\textsc{Literal-Self-Assn}}\label{tag:lit-self-assn}\\
       & \set{s_1,\cdots,s_k} && \mbox{ if }\left\{\begin{array}{l}1.\AssnMap(s)=\bigwedge_k \phi_k~\mbox{\textbf{or}}~ \AssnMap(s)=\bigvee_k \phi_k,~\\2. \AssnMap(s_k)=\phi_k,~\\3. \forall k, \ResdMap(s)=\ResdMap(s_k), \\4. op=\epsilon \end{array}\right.\tag{\textsc{Compound-Assn}}\label{tag:compound-assn}\\
       & \{\} && \mbox{ otherwise }\\
      \tikzmark{snddd}{} & &&
    \end{align*}
  \end{subequations}
}}
\vspace{-0.5cm}
\caption{Transition function used in the Definition \ref{def:afa}}
\vspace{-0.5cm}
\label{fig:afa:nextst}
\end{figure}
Following Point \ref{create}, any state added to $S$ is either annotated with a smaller $\ResdMap$ or a smaller formula compared to the states already present in $S$. Further, every formula and trace $\sigma$ is of finite length. Hence the set of states $S$ is finite. By Point \ref{uniexist} of this construction, a state $s$ where $\AssnMap(s)$ is a compound formula, is always a universal state irrespective of whether $\AssnMap(s)$ is a conjunction or a disjunction of clauses. The reason behind this decision will be clear shortly when we will use this AFA to inductively construct the weakest precondition $\weakestpre{\subst{\sigma}{assume}{assert}}{\phi}$. Note that we assume every formula is normalized in CNF.
Figure \ref{fig:peterson:afatrace} shows an example trace $\sigma=\lb{abApqPrcs}$ of Peterson's algorithm. This trace is picked from the Peterson's specification in Figure \ref{fig:petersonpro}. To prove $\sigma$ correct with respect to the safety formula $\phi\defeq(\ell_2=2)$ we first construct $\afa{\run,\neg \phi}$ which will later help us to derive $\weakestpre{\subst{\sigma}{assume}{assert}}{\neg \phi}$. This AFA is shown in Figure \ref{fig:peterson:afa}. For a state $s$, $\AssnMap(s)$ is written inside the rectangle representing that state and $\ResdMap(s)$ is written inside an ellipse next to that state. We show here some of the steps illustrating this construction.
\begin{figure*}[t]
\removelatexerror
 \begin{minipage}[b]{0.6\linewidth}
\hspace{0.5cm}
 \scalebox{0.65}{\parbox{\linewidth}{%
 \stepninehorizononly{blue!20}}}
 \caption{AFA of trace given in Figure \ref{fig:peterson:afa}(b) and $\phi=\neg (\ell_2=2)$}
 \label{fig:peterson:afa}
 \end{minipage}
\hspace{0.6cm}
\begin{minipage}[b]{0.35\linewidth}
  \scalebox{0.8}{\parbox{\linewidth}{%
 \hspace{1cm}
 $\begin{array}{l}
\lb{a.~} \Assign{flag_1}{\true}\\
\lb{b.~} \Assign{turn}{2}\\
\lb{A.~} \assume{\begin{array}{l}\mathrm{flag_2=\false}~ \\ ||~ \mathrm{turn=1}\end{array}}\\
\lb{p.~} \Assign{flag_2}{\true}\\
\lb{q.~} \Assign{turn}{1}\\
\lb{P.~} \assume{\begin{array}{l}\mathrm{flag_1=\false}~ \\||~ \mathrm{turn=2}\end{array}}\\
\lb{r.~} \Assign{res}{2}\\
\lb{c.~} \Assign{res}{1}\\
\lb{s.~} \Assign{\ell_2}{res}
\end{array}$}}
 \vspace{1cm}
 \caption{A trace from Peterson's algorithm }
 \label{fig:peterson:afatrace}
 \end{minipage}
\vspace{-0.2cm}
\end{figure*}
\begin{enumerate}
 \item By Definition \ref{def:afa}, we have $\AssnMap(s_0)=\neg (\ell_2=2)$ and $\ResdMap(s_0)=\sigma=\lb{abApqPrcs}$ for initial state $s_0$.
 \item In a transition $\delta(s,op)=\{s'\}$ created by Rule \ref{tag:lit-assn} the state $s'$ is annotated with the weakest precondition of an operation $op$, taken from $\ResdMap(s)$, with respect to $\AssnMap(s)$. Operation $op$ is picked in such a way that  $\AssnMap(s)$ is \emph{stable} with respect to every other operation present after $op$ in $\ResdMap(s)$. Such transitions capture the inductive construction of the weakest precondition for a given $\phi$ and trace $\sigma$. Transition $\delta(s_0,\lb{s})=\{s_1\}$ in Figure \ref{fig:peterson:afa} is created by this rule as $\weakestpre{\subst{\lb{s}}{assume}{assert}}{\AssnMap(s_0)}=\AssnMap(s_1)$, and $\ResdMap(s_0)=\ResdMap(s_1).\lb{s}$. 
 \item  In any transition created by Rule \ref{tag:compound-assn}, say from $s$ to $s_1,\cdots,s_k$, the states $s_1,\cdots s_k$ are annotated with the subformulae of $\AssnMap(s)$. 
   For example, transitions $\delta(s_3,\epsilon)=\{s_4,s_5\}$ and $\delta(s_7,\epsilon)=\{s_8,s_9\}$.
 \item Transition $\delta(s_8,\lb{a})=\{s_{12}\}$ follows from the rule \ref{tag:lit-assn}. Note that $\ResdMap(s_{12})$ is empty and hence by Point \ref{fin} of Definition \ref{def:afa}, $s_{12}$ is an accepting state. 
 Following the same reasoning, states $s_6$, $s_{10}$ and $s_{13}$ are also set as accepting states.
 \item  Rule \ref{tag:lit-self-assn} adds a self transition at a state $s$ on a symbol $op\in \mathcal{OP}_{\epsilon}$ such that  $\AssnMap(s)$ is \emph{stable} with respect to $\subst{op}{assume}{assert}$. For example, transitions $\delta(s_0,op)=\{s_0\}$ where $op\in \mathcal{OP}_{\epsilon}\setminus\{\lb{s,A,P}\}$. 
 \end{enumerate}

 The following lemma relates $\ResdMap(s)$ at any state to the set of words accepted by $s$ in this AFA.
\begin{lemma}\label{lem:runinafa}
 Given a $\run\in \lang{\aut{\mathcal{P}}}$ and $\phi$, let $\afa{\run,\phi}$ be the AFA satisfying Definition \ref{def:afa}.
 For every state $s$ of this AFA, the condition $\rev(\ResdMap(s))\in \acc(s)$ holds.
\end{lemma}
A detailed proof of this lemma is given in Appendix {\ref{prf:lemruninafa}}. 
This lemma uses the reverse of $\ResdMap(s)$ in its statement because the weakest precondition of a sequence is constructed by scanning it from the end. This can be seen in the transition rule \ref{tag:lit-assn}. As a corollary, $\rev(\sigma)$ is also accepted by this AFA because by Definition \ref{def:afa}, $\ResdMap(s_0)$ is $\sigma$.
\qed

\begin{figure}
\hspace{0.4cm}\scalebox{0.75}{\parbox{\linewidth}{%
\begin{subequations}  
$~~~\HMap(s)=$
\begin{align*}
       \tikzmark{fsttt}{} & \AssnMap(s) & \mbox{ if }s\in S_F \tag{\textsc{Base-case}}\label{tag:base}\\
       & \bigwedge_k \HMap(s_k) & \mbox{ if } \delta(s,\epsilon)=\{s_1,\cdots,s_k\} 
	\mbox{ and } \AssnMap(s)=\bigwedge_k \AssnMap(s_k) \tag{\textsc{Conj-case}}\label{tag:conj}\\
	&
        \bigvee_k \HMap(s_k) & \mbox{ if } \delta(s,\epsilon)=\{s_1,\cdots,s_k\}
	      \mbox{ and } \AssnMap(s)=\bigvee_k \AssnMap(s_k) \tag{\textsc{Disj-case}}\label{tag:disj}\\ 
      \tikzmark{snddd}{} & \HMap(s')  &\mbox{ if } (s,op,\{s'\})\in \delta \tag{\textsc{Lit-case}}\label{tag:prop}
    \end{align*}
  \end{subequations}
\EmBracet[black]{fsttt.north west}{snddd.west}
}}
\caption{Rules for $\HMap$ construction}
\vspace{-0.5cm}
\label{fig:hmap}
\end{figure}
\subsection{Constructing the weakest precondition from $\afa{\run,\phi}$} 
After constructing $\afa{\run,\phi}$ the rules given in Figure \ref{fig:hmap} are used to inductively construct and assign a formula, $\HMap(s)$, to every state $s$ of $\afa{\run,\phi}$. 
Figure \ref{fig:petex:wp} shows the AFA of Figure \ref{fig:peterson:afa} where states are annotated with formula $\HMap(s)$. This formula is shown in the ellipse beside every state. For better readability we do not show $\ResdMap(s)$ in this figure. 

Following Rule \ref{tag:base}, $\HMap$ of $s_6,s_{12}$, and $s_{13}$ are set to $\false$ whereas $\HMap(s_{10})$ is set to $\mathrm{flag_2=\false}$. By Rule \ref{tag:prop}, $\HMap$ of $s_5,s_8$ and $s_{11}$ are also set to $\false$. After applying Rule \ref{tag:disj} for transition $\delta(s_9,\epsilon)=\{s_{10},s_{11}\}$, $\HMap(s_9)$ is set to $\mathrm{flag_2=\false}$. Similarly, using Rule \ref{tag:conj} we get $\HMap(s_7)$ as $\false$. Finally, $\HMap(s_0)$ is also set to $\false$.
%
\begin{figure*}
 \removelatexerror
 \begin{minipage}[b]{0.6\linewidth}
 \hspace{-0.5cm}
 \scalebox{0.65}{\parbox{\linewidth}{%
 \stepninehorizon{blue!20}}}
 \caption{$\HMap$ construction for the running example}
  \label{fig:petex:wp}
 \end{minipage}
 \hspace{-1cm}
 \begin{minipage}[b]{0.5\linewidth}
   \scalebox{0.75}{\parbox{\linewidth}{%
  \begin{algorithm2e}[H]
\SetAlgoLined
\SetKwFunction{isReady}{isReady}
\KwData{Input AFA $\langle \faS,\teS,\mathcal{OP}\cup \set{\epsilon},s_0,S_F,\AssnMap,\ResdMap\rangle$}
\KwResult{Modified AFA}
Let $s$ be a state in AFA such that $s\in \faS$, $\delta(s,\epsilon)=\{s_1,\cdots,s_k\}$, $\HMap(s)$ is unsatisfiable, and $\AssnMap(s)=\bigwedge_k \AssnMap(s_k)$\;
Let $\unsatcore(s) \subseteq \Powerset{\{s_1,\cdots,s_k\}}$ such that $\{s'_1,\cdots,s'_n\} \in \unsatcore(s)$ $\textrm{iff}$ $\set{\HMap(s'_1),\cdots,\HMap(s'_n)}$ is a minimal unsat core of $\bigwedge_k \HMap(s_k)$ \label{mod1}\;
Create an empty set $\mathtt{U}$\;
\ForEach{$\{s'_1,\cdots,s'_n\}\in \unsatcore(s)$}{\label{modloop}
    create a new universal state $s_u \in \faS$ and add it to the set $\mathtt{U}$\label{mod3}\;
    Set $\AssnMap(s_u)=\bigwedge_i\AssnMap(s'_i)$ \label{mod4}\;
    Set $\HMap(s_u)=\bigwedge_i\HMap(s'_i)$\label{mod5}\;
    Add a transition by setting $\delta(s_u,\epsilon)=\{s'_1,\cdots,s'_n\}$\label{mod6}\;
  }
Remove transition $\delta(s,\epsilon)=\{s_1,\cdots,s_k\}$\;
Convert $s$ to an existential state\;
Add a transition from $s$ on $\epsilon$ by setting $\delta(s,\epsilon)=\mathtt{U}$ where $\mathtt{U}$ is the set of universal states created one for each element of $\unsatcore(s)$\label{mod10}\;
\caption{Converting universal to existential states while preserving Lemma \ref{lem:hoare}}
\label{fig:afamodification}
\end{algorithm2e}}}
\vspace{0.5cm}
 \end{minipage}
 \vspace{-1cm}
\end{figure*}
$\HMap$ constructed inductively in this manner satisfies the following property;
\begin{lemma}\label{lem:hoare}
  Let $\afa{}$ be an AFA constructed from a trace and a post condition as in Definition \ref{def:afa} then for every state $s$ of this AFA and for every word $\sigma$ accepted by state $s$, $\HMap(s)$ is logically equivalent to $\weakestpre{\subst{\rev(\sigma)}{assume}{assert}}{\AssnMap(s)}$.
\end{lemma}
Here we present the proof outline. Detailed proof is given in Appendix \ref{prf:lemhoare}.
First consider the accepting states of $\afa{}$. For example, states $s_6$, $s_{10}$, $s_{12}$ and $s_{13}$ of Figure \ref{fig:petex:wp}. Following the definition of an accepting state and by the self-loop adding transition rule \ref{tag:lit-self-assn}, every word $\sigma$ accepted by such an accepting state $s$ satisfies $\weakestpre{\subst{\rev(\sigma)}{assume}{assert}}{\AssnMap(s)}=\AssnMap(s)$. Therefore, setting $\HMap(s)$ as $\AssnMap(s)$ for these accepting states, as done in Rule \ref{tag:base} completes the proof for accepting states.

Now consider a state  $s$ with transition $\delta(s,\epsilon)=\{s_1,\cdots,s_k\}$, created using Rule \ref{tag:compound-assn}, and let $\sigma$ be a word accepted by $s$. By construction, $s$ must be a universal state and hence $\sigma$ must be accepted by each of $s_1,\cdots,s_k$ as well. Using this lemma inductively on successor states $s_1,\cdots,s_k$ (induction on the formula size) we get $\weakestpre{\subst{\sigma}{assume}{assert}}{\AssnMap(s_i)}=\HMap(s_i)$ for all $i\in\{1\cdots k\}$. Now we can apply Property \ref{prop:wp} depending on whether $\AssnMap(s)$ is a conjunction or a disjunction of $\AssnMap(s_k)$. By replacing $\AssnMap(s)$ with $\bigvee_k \AssnMap(s_k)$($\bigwedge_k \AssnMap(s_k)$) and $\HMap(s)$ with $\bigvee_k \HMap(s_k)$($\bigwedge_k \HMap(s_k)$)
completes the proof. Note that, making $s$ as a universal state when $\AssnMap(s)$ is either a conjunction or a disjunction allowed us to use Property \ref{prop:wp} in this proof. Otherwise, if we make $s$ an existential state when $\AssnMap(s)$ is a disjunction of formulae then we can not prove this lemma for states where $\HMap(s)$ is constructed using Rule \ref{tag:disj}.
\qed 

This lemma serves two purposes. First, it checks the correctness of a trace $\sigma$ w.r.t. a safety property for which this AFA was constructed. If $\HMap(s_0) \lland \mathcal{I}$ is unsatisfiable, as in our Peterson's example trace, then $\sigma$ is declared as correct. Second, it guarantees that every trace accepted by this AFA, that is present in the set of all traces of $\mathcal{P}$, is also safe and hence we can skip proving their correctness altogether. Removing such traces is equivalent to subtracting the language of this AFA from the language representing the set of all traces. Then a natural question to ask is if we can increase the set of accepted words of this AFA while preserving Lemma \ref{lem:hoare}.

\subsection{Enlarging the set of words accepted by $\afa{\run,\phi}$}\label{subsec:trans}
\begin{figure*}[t]
 \removelatexerror
\begin{minipage}{\linewidth}
\begin{minipage}[b]{0.2\linewidth}
  \scalebox{0.7}{\parbox{\linewidth}{%
  \centering{\hspace{1cm}
  \exafa{1}}}}
\caption{Example Trace}
\label{fig:trans1ex}
\end{minipage}
\hspace{-.75cm}
\begin{minipage}[b]{0.4\linewidth}
  \scalebox{0.65}{\parbox{\linewidth}{%
	    \exafaoneannotated{blue!20}}}
	  \caption{AFA for $\sigma$ given in Figure \ref{fig:trans1ex}}
	  \label{fig:trans1afa}
\end{minipage}
\hspace{-1cm}
\begin{minipage}[b]{0.4\linewidth}
\hspace{-1cm}
\scalebox{0.65}{\parbox{\linewidth}{%
  \stepninehorizonhmapextra}}
 \caption{AFA of Figure \ref{fig:petex:wp} after Modification}
  \label{fig:petex:redraw}
 \end{minipage}
\end{minipage}
\end{figure*}
\IEEEPARstart{\textbf{Converting Universal States to Existential States }}
Figure \ref{fig:trans1ex} shows an example trace $\sigma=\lb{abcde}$ obtained from the parallel composition of some program $P$.
Figure \ref{fig:trans1afa} shows the AFA constructed for $\sigma$ and $\phi$ as $S<t \lland z<x$. From Lemma \ref{lem:hoare} we get $\weakestpre{\sigma}{\phi}$ as $\false$. Note that the $\weakestpre{\sigma}{S<t}$ and $\weakestpre{\sigma}{z<x}$ are unsatisfiable, i.e. we have two ways to derive the unsatisfiability of $\weakestpre{\sigma}{\phi}$;
one is due to the operation $\lb{d}$,
and the other is due to the operation $\lb{a}$ followed by operation $\lb{e}$.
In this example, any word that enforces either of these two ways will derive $\false$  as the weakest precondition.
For example, the sequence $\run'=\lb{adcbe}$ is not accepted by the AFA of Figure \ref{fig:trans1afa}
but the condition $\weakestpre{\rev(\sigma')}{\neg \phi}=\false$ follows from $\weakestpre{\lb{d}}{\neg \phi}=\false$ 
which is already captured in the AFA of Figure \ref{fig:trans1afa}. 
Note that states $s_1$ and $s_2$ in Figure \ref{fig:trans1afa} are annotated with unsatisfiable $\HMap$ assertion.
It seems sufficient to take any one of these branches to argue the unsatisfiability of $\HMap(s_0)$ because $\HMap(s_0)$, by definition, is a \emph{conjunction} of $\HMap(s_1)$ and $\HMap(s_2)$. Therefore, if we convert $s_0$, a universal state, to an existential state then the modified AFA will accept $\lb{adcbe}$. Let us look at 
Algorithm \ref{fig:afamodification} to see the steps involved in this transformation. This algorithm picks a universal state $s$ such that $\AssnMap(s)$ is a \emph{conjunction} of clauses and only a subset of its successors are sufficient to make $\HMap(s)$ unsatisfiable. State $s_0$ of Figure \ref{fig:trans1afa} is one such state. For each such minimal subsets of its successors, this algorithm creates a universal state, as shown in Line \ref{mod3} of this algorithm. It is easy to see that $\HMap(s_u)$ is also unsatisfiable. Before adding $\delta(s_u,\epsilon)=\{s'_1,\cdots,s'_n\}$ transition in AFA this algorithm sets $\AssnMap(s_u)$ as $\bigwedge_i \AssnMap(s'_i)$. By construction, every word accepted by $s_u$ must be accepted by $s'_1,\cdots,s'_n$. Each of these states $s'_1,\cdots,s'_n$ satisfy Lemma \ref{lem:hoare}. Hence Lemma \ref{lem:hoare} continues to hold for these newly created universal states as well.
Now consider a newly created transition $(s,\epsilon,\mathtt{U})$ in Line \ref{mod10}. For any state $s''\in \mathtt{U}$, $\AssnMap(s)$ logically implies $\AssnMap(s'')$ because $s''$ represents a subset of the original successors of $s$, viz. $s_1,\cdots,s_k$. As $s$ is now an existential state, any word accepted by $s$, say $\sigma'$, is accepted by at least one state in $\mathtt{U}$, say $s'$. Using Lemma \ref{lem:hoare} on $s'$, $\HMap(s')$ is logically equivalent to $\weakestpre{\subst{\rev(\sigma')}{assume}{assert}}{\AssnMap(s')}$. Using unsatisfiability of $\HMap(s)$ and $\HMap(s')$ and the monotonicity property of the weakest precondition, Property \ref{prop:wpsubset}, we get that $\HMap(s)$ is logically equivalent to $\weakestpre{\subst{\rev(\sigma')}{assume}{assert}}{\AssnMap(s)}$. 
This transformation is formally proved correct in Appendix \ref{prf:modification1hoare}.

\IEEEPARstart{\textbf{Adding More transitions to $\afa{\run,\phi}$ using the Monotonicity Property of the Weakest Precondition }}
We further modify $\afa{\run,\phi}$ by adding more transitions. For any two states $s$ and $s'$ such that $\AssnMap(s) $ and $\AssnMap(s')$ are literals, both $\HMap(s)$ and $\HMap(s')$ are unsatisfiable, and there exists a symbol $a$ (can be $\epsilon$ as well) such that $\weakestpre{\subst{a}{assume}{assert}}{\AssnMap(s)}$ logically implies $\AssnMap(s')$, an edge labeled $\mathrm{a}$ is added from $s$ to $s'$. This transformation also preserves Lemma \ref{lem:hoare} following the same monotonicity property, Property \ref{prop:wpsubset} used in the previous transformation. Similar argument holds when $\HMap(s)$ and $\HMap(s')$ are valid and $\AssnMap(s')\limp \weakestpre{a}{\AssnMap(s)}$ holds. The rules of adding edges are shown in Figure \ref{fig:addedge}.
\begin{figure}
\scalebox{0.8}{\parbox{\linewidth}{%
\begin{subequations}  
$ \delta(s,op)=\delta(s,op)\cup\{s'\} \mbox{ iff }$\\
\begin{align*}
       \tikzmark{fsttt}{} &  \HMap(s)\mbox{ and } \HMap(s')\mbox{ are unsatisfiable,} \\
       & (s)\mbox{ is a literal, and}\\
       &{\subst{op}{assume}{assert}}{\AssnMap(s)}\limp \AssnMap(s')\tag{\textsc{Rule-Unsat}} \label{tag:unsatrule}\\
  & \mathbf{OR}\\
	&  \HMap(s)\mbox{ and } \HMap(s')\mbox{ are valid}\\
	&(s)\mbox{ is a literal, and}\\
      \tikzmark{snddd}{} &(s')\limp \weakestpre{\subst{op}{assume}{assert}}{\AssnMap(s)}\tag{\textsc{Rule-Valid}} \label{tag:validrule}
    \end{align*}
  \end{subequations}
\EmBracet[black]{fsttt.north west}{snddd.west}
}}
\caption{Rules for adding more edges}
\vspace{-0.5cm}
\label{fig:addedge}
\end{figure}

Figure \ref{fig:petex:redraw} shows the AFA of Figure \ref{fig:petex:wp} modified by above transformations. Rule \ref{tag:unsatrule} adds an edge from $s_4$ to $s_8$ on symbol $\epsilon$ because $\HMap(s_4)$ and $\HMap(s_8)$ are unsatisfiable and $\weakestpre{\epsilon}{\AssnMap(s_4)}$ logically implies $\AssnMap(s_8)$. Same rule also 
adds a self loop at $s_8$ on operation $\lb{P}$ and a self loop at $s_2$ on operation $\lb{A}$.  
Transformation by Algorithm \ref{fig:afamodification} removes the transition from $s_7$ to $s_{9}$ and all other states reachable from $s_{9}$. 
Now consider a trace $\rev(\lb{abpqPArcs})$ that is accepted by this modified AFA in Figure \ref{fig:petex:redraw} but was not accepted by the original AFA of Figure \ref{fig:petex:wp}. Note that $\weakestpre{\lb{abpqPArcs}}{\neg(\ell_2=2)}$ is unsatisfiable and this is a direct consequence of Lemma \ref{lem:hoare}. Because of the transformations presented in this sub-section we do not need to reason about this trace separately. 
This transformation is formally proved correct in Appendix \ref{prf:modification2hoare}.
\begin{figure}[t]
 \removelatexerror
 \hspace{.7cm}
  \scalebox{0.75}{\parbox{\linewidth}{%
\begin{algorithm2e}[H]
\SetAlgoLined
\SetKwFunction{isReady}{isReady}
\KwIn{A concurrent program $\mathcal{P}=\set{p_1,\cdots,p_n}$ with safety property map $\overline{\AssnM}$}
\KwResult{$yes$, if program is safe else a counterexample}
Let $\aut{\mathcal{P}}$ bet the automaton that represents the set of all the SC executions of $P$ (as defined in Section \ref{sec:prelim})\;
Set {$\mathtt{tmp}:=\mathcal{L}(\htt{\A{P}})$}\;
\While{$\mathtt{tmp}$ is not empty}{\label{finalloop}
 \label{four} Let $\run \in \mathtt{tmp}$ with $\phi$ as a safety assertion to be checked\;
 Let $\afa{\run,\neg \phi}$ be the AFA constructed from $\run$ and $\neg \phi$ \label{finalafa}\;
 \uIf{$\mathcal{I} \lland \HMap(s_0)$ is satisfiable}{\label{finalcounter}
  $\run$ is a valid counterexample violating $\phi$\;
  \KwRet($\run$)\label{finalret1}\;
  }\Else{
  Let {$\afa{}'$} be the AFA modified by proposed transformations\;
  {${\mathtt{tmp}}:=\mathtt{tmp}\setminus Rev$}\label{langsub},  where $Rev=\{\rev(\sigma)\mid \sigma\in \mathcal{L}(\afa{}')\}$\;
  }
}
\KwRet($yes$)\label{finalret2}\; 
\caption{Algorithm to check the safety assertions of a concurrent program $P$}
\label{fig:afacombined}
\end{algorithm2e}}}
\vspace{-0.5cm}
\end{figure}
\subsection{Putting All Things Together For Safety Verification}
In Algorithm \ref{fig:afacombined}, all the above steps are combined to check if all the SC executions of a concurrent program $P$ satisfy the safety properties specified as assertions. 
Proof of the following theorem is given in Appendix \ref{ref:thmtermination}.
\begin{theorem}\label{lem:termination}
 Let $P=(p_1,\cdots,p_n)$ be a finite state program (with or without loops) with associated assertion maps $\AssnM_{p_i}$. All assertions of this program hold $\textrm{iff}$ Algorithm \ref{fig:afacombined} returns $yes$. If the algorithm returns a word $\run$ then at least one assertion fails in the execution of $\run$.
\end{theorem}

\section{Experimental Evaluation}\label{sec:experiments}
We implemented our approach in a prototype tool,
\texttt{ProofTraPar}. This tool reads the input program written in a
custom format. In future, we plan to use off-the-shelf parsers such as CIL or
LLVM to remove this dependency. Individual processes are represented
using finite state automata. 
We use an
automata library, libFAUDES \cite{libfaudes} to carry out operations
on automata. As this library does not provide operations on AFA,
mainly complementation and intersection, we implemented them in our
tool. After constructing the AFA from a trace we first remove
$\epsilon$ transitions from this AFA. This is followed by adding
additional edges in AFA using proposed
transformations. Instead of reversing this AFA (as in Line \ref{langsub}
of Algorithm \ref{fig:afacombined}) we subtract it with an NFA that
represents the reversed language of the set of all traces. This avoids
the need of reversing an AFA. Note that we do not convert our AFA to
NFA but rather carry out intersection and complementation operations (needed for language subtraction operation)
directly on AFA.
Our tool uses the Z3 \cite{z3} theorem prover to check the validity
of formulae during AFA construction. \texttt{ProofTraPar} can be
accessed from the repository
\url{https://github.com/chinuhub/ProofTraPar.git}.

Figure \ref{tab:res} tabulates the result of verifying
\textit{pthread-atomic} category of SV-COMP benchmarks using our tool,
THREADER \cite{Gupta:2011:TCV:2032305.2032337} and Lazy-CSeq
\cite{Lazy-CSeq}.
These tools were the winners in the concurrency category of the
software verification competition of 2013 (THREADER), 2014 and 2015
(Lazy-CSeq).
Dash (--) denotes that the tool did not finish the analysis within 15
minutes. Numbers in bold text denote the best time of that
experiment. Safe/Unsafe versions of these programs are labeled with
\textit{.safe}/\textit{.unsafe}. Except on Reader-Writer Lock and on
unsafe version of QRCU(Quick Read Copy Update), our tool performed
better than the other two tools. On unsafe versions, our approach took
more time to find out an erroneous trace as compared to Lazy-CSeq
\cite{Lazy-CSeq}. Context-bounded exploration by Lazy-CSeq
\cite{Lazy-CSeq} and the presence of bugs at a shallow depth seem to be a
possible reason behind this performance difference. Introducing priorities while picking traces in order to 
make our approach efficient in bug-finding is left open for future work.
\begin{figure}
 \scalebox{0.75}{\parbox{\linewidth}{%
 \hspace{-0.5cm}
\begin{tabular}{l|c|c|c}
Program & ProofTraPar & THREADER\cite{Gupta:2011:TCV:2032305.2032337} & Lazy-CSeq\cite{Lazy-CSeq}\\
\hline& & &   \\
Peterson.safe & \textbf{0.3} & 3.2 & 3.1\\ 
Dekker.safe & \textbf{1.1} & 1.7 & 4.2 \\
Lamport.safe & \textbf{2.4} & 47 & 5.1 \\
Szymanksi.safe & \textbf{3} & 12.8 & 4\\
TimeVarMutex.safe & \textbf{0.76} & 8.56 & 4.2 \\
RWLock.safe (2R+2W) & 8.8 & 140 & \textbf{6.7}\\
RWLock.unsafe (2R+2W)  & 3.8 & 153 & \textbf{0.7}\\
Qrcu.safe (2R+1W) & \textbf{20} & -- & 41 \\
Qrcu.unsafe (2R+1W) & 13.8 & 76 & \textbf{1.1}\\
\end{tabular}
}}
\caption{\scriptsize{Comparison with THREADER\cite{Gupta:2011:TCV:2032305.2032337}, and Lazy-CSeq \cite{Lazy-CSeq}} (Time in seconds)}
\label{tab:res}
\vspace{-0.5cm}
\end{figure}
\section{Related Work}\label{sec:related}
Verifying the safety properties of a concurrent program is a well studied area. 
Automated verification tools which use model checking based approaches employ optimizations such as Partial Order Reductions (POR) \cite{Peled:1993:OOM:647762.735490,Godefroid:1996:PMV:547238,Flanagan:2005:DPR:1040305.1040315} to handle larger number of interleavings. These optimizations also selectively check a representative set of traces among the set of all interleavings. POR based methods were traditionally used in bug finding but recently they have been extended efficiently, using abstraction and interpolants, for proving programs correct \cite{wko2013}. The technique presented in this paper, using AFA, can possibly be used to keep track of partial orders in POR based methods.
In \cite{ctp}, a formalism called concurrent trace program (CTP) is defined to capture a set of interleavings corresponding to a concurrent trace. CTP captures the partial orders encoded in that trace. Corresponding to a CTP, a formula $\phi_{ctp}$ is defined such that $\phi_{ctp}$ is satisfiable iff there is a feasible linearization of the partial orders encoded in CTP that violates the given property. Our AFA is also constructed from a trace but unlike CTP it only captures those different interleavings which guarantee the same proof outline.
Recently in \cite{Gupta:2015:SRC:2775051.2677008}, a formalism called \emph{HB-formula} has been proposed 
to capture the set of happens-before relations in a set of executions. 
This relation is then used for multiple tasks such as synchronization synthesis\cite{DBLP:conf/cav/CernyCHRRST15}, bug summarization and predicate refinement. 
Since the AFA constructed by our algorithm can also be represented as a boolean formula (universal states correspond to conjunction and existential states correspond to disjunction) that encodes the ordering relations among the participating events, it will be interesting to explore other usages of this AFA along the lines of \cite{Gupta:2015:SRC:2775051.2677008}.
%


\section{Conclusion and Future Work}\label{sec:conclude}
We presented a trace partitioning based approach for verifying
safety properties of a concurrent program.  To this end, we 
introduced a novel construction of an alternating finite automaton to
capture the proof of correctness of a trace in a program.  We also presented an implementation of our 
algorithm which compared competitively with existing state-of-the-art tools. 
We plan to extend this approach for
parameterized programs and programs under relaxed memory models. 
We also plan to investigate the use of
interpolants with weakest precondition axioms to incorporate
abstraction for handling infinite state programs.

\bibliography{paper}
\bibliographystyle{plain}
\newpage
\appendix
\pagenumbering{arabic}
\section{Proofs of the Paper}

\subsection{Proof of Lemma \ref{lem:assumeassrtrhs2}}\label{prf:lemassumeassrtrhs2}
We prove it by induction on $n$.
\begin{enumerate}
 \item Base case $\len{\sigma}=0$: If $\len{\sigma}=0$ then $\weakestpre{\subst{\run}{assume}{assert}}{\neg \phi} =\neg \phi$. If $\neg \phi \lland \mathcal{I}$ is unsatisfiable then $\mathcal{I}$ satisfies $\phi$. Hence proved.
  \item Induction step, $\len{\sigma}=n+1$: Let $\sigma=\sigma'.a$. If $\weakestpre{\subst{\run'.a}{assume}{assert}}{\neg \phi} \lland \mathcal{I}$ is unsatisfiable then following cases can happen based on $a$.
    \begin{itemize}
     \item $a: x:=E$:- If $\weakestpre{\subst{\run'.a}{assume}{assert}}{\neg \phi} \lland \mathcal{I}$ is unsatisfiable then $\weakestpre{\subst{\run'}{assume}{assert}}{\weakestpre{a}{\neg \phi}} \lland \mathcal{I}$ is also unsatisfiable. By substituting $\weakestpre{a}{\neg \phi}$ with $\subst{\neg \phi}{E}{x}$ we get that $\weakestpre{\subst{\run'}{assume}{assert}}{\subst{\neg \phi}{E}{x}} \lland \mathcal{I}$  is unsatisfiable. Using IH on $\sigma'$ it implies that after executing $\sigma'$ from $\mathcal{I}$ the resultant state either does not terminate or terminates in a state satisfying $\subst{\phi}{E}{x}$. If $\sigma'$ does not terminate then so does the execuction of $\sigma$ starting from $\mathcal{I}$. If $\sigma'$ terminates in a state satisfying $\subst{\phi}{E}{x}$ then by the definition of the weakest precondition, execution of $a$ from this state will satisfy $\phi$. Hence proved.
     \item $a: \assume{\phi'}$:-If $\weakestpre{\subst{\run'.a}{assume}{assert}}{\neg \phi} \lland \mathcal{I}$ is unsatisfiable then $\weakestpre{\subst{\run'}{assume}{assert}}{\weakestpre{a}{\neg \phi}} \lland \mathcal{I}$ is also unsatisfiable. By substituting $\weakestpre{a}{\neg \phi}$ with $\phi' \lland \neg \phi$ we get that $\weakestpre{\subst{\run'}{assume}{assert}}{\phi' \lland \neg \phi} \lland \mathcal{I}$  is unsatisfiable. Using IH on $\sigma'$ it implies that after executing $\sigma'$ from $\mathcal{I}$ the resultant state either does not terminate or terminates in a state satisfying $\neg \phi \lor \phi'$. If $\sigma'$ does not terminate then the execution of $\sigma$ from $\mathcal{I}$ does not terminate as well. If $\sigma'$ terminates in a state satisfying $\neg \phi$ then the execution of $a$ blocks and hence the execution of $\sigma$ does not terminate. If $\sigma'$ terminates in a state satisfying $\phi'$ but $\neg \phi$ does not hold then $\phi \land \phi'$ must hold. Execution of $\assume{\phi'}$  acts as $\nop$ instruction and the resultant state satisfies $\phi$. hence proved.
     \item $a: \lcas{x}{v_1}{v_2}$:- As weakest precondition of $\lcas{x}{v_1}{v_2}$ is obtained from the weakest precondition of assignment and assume instruction hence the similar reasoning works for this case.
    \end{itemize}

\end{enumerate}

\subsection{Proof of Lemma \ref{lem:assumeassrtlhs2}}\label{prf:lemassumeassrtlhs2}

\begin{proof}
 Let us prove it by induction on the length of $\run$.\\
 \begin{enumerate}
  \item Base case, $\len{\run}=0$: When the length of $\run$ is 0 and $\mathcal{I} \lland \neg \phi$ is satisfiable then $\mathcal{I}$ does not satisfy $\phi$. Hence proved.
  \item Induction Step, $\len{\run}=n+1$: Let $\sigma=\sigma'.a$. Following case can happen based on the type of $a$.
  \begin{itemize}
   \item $a: x:=E$:- If $\weakestpre{\subst{\run}{assume}{assert}}{\neg \phi} \lland \mathcal{I}$ is satisfiable then $\weakestpre{\subst{\run'}{assume}{assert}}{\weakestpre{a}{\neg \phi}} \lland \mathcal{I}$ is also satisfiable. By substituting $\weakestpre{a}{\neg \phi}=\subst{\neg \phi}{E}{x}$ we get that $\weakestpre{\subst{\run'}{assume}{assert}}{\subst{\neg \phi}{E}{x}} \lland \mathcal{I}$ is satisfiable. By IH on $\run'$, execution of $\sigma'$ from $\mathcal{I}$ terminates in a state not satisfying $\subst{\phi}{E}{x}$. By definition of the weakest precondition, the state reached after executing $a$ from this state does not satisfy $\phi$. Hence proved.
   \item $a: \assume{\phi'}$:-If $\weakestpre{\subst{\run}{assume}{assert}}{\neg \phi} \lland \mathcal{I}$ is satisfiable then $\weakestpre{\subst{\run'}{assume}{assert}}{\weakestpre{\subst{\assume{\phi'}}{assume}{assert}}{\neg \phi}} \lland \mathcal{I}$ is also satisfiable. By substituting $\weakestpre{\subst{\assume{\phi'}}{assume}{assert}}{\neg \phi}=\phi'\land \neg \phi$ we get that $\weakestpre{\subst{\run'}{assume}{assert}}{\neg (\neg \phi' \llor \phi)} \lland \mathcal{I}$ is satisfiable. By IH on $\run'$, execution of $\sigma'$ from $\mathcal{I}$ terminates in a state not satisfying $\neg \phi' \llor \phi$. In other words, $\phi'$ and $\neg \phi$ holds in the state reached after executing $\sigma'$ from $\mathcal{I}$. Therefore, after executing $\assume{\phi'}$, the resultant state satisfies $\neg \phi$ and hence proved.
   \item $a: \lcas{x}{v_1}{v_2}$:-Similar to the combination of above two cases.
  \end{itemize}
 \end{enumerate}

\end{proof}

\subsection{Proof of Lemma \ref{lem:runinafa}}\label{prf:lemruninafa}
\begin{proof}
We use induction for this proof. Let us use the following ordering on the states of $\afa{\run}{\phi}$. For any two states $s$ and $s'$, $s<s'$  if $\len{\ResdMap(s)}<\len{\ResdMap(s')}$ or if lengths are same then $\AssnMap(s)$ is a sub formula of $\AssnMap(s')$. Any two states which are not related by this order, put them in any order to make $<$ as a total order. It is clear that the smallest state in this total order must be one of the accepting state. Now we are ready to proceed by induction using this total order.

\begin{itemize}
 \item Base case; For every accepting state $s\in S_F$, by Point \ref{fin} of Definition \ref{def:afa}, the condition $\weakestpre{op}{\AssnMap(s)}=\AssnMap(s)$ holds for every $op\in \Elem{\ResdMap(s)}$. Further, By transition rule \ref{tag:lit-self-assn} of this AFA, a self transition must be there for all such $op\in \Elem{\ResdMap(s)}$ and hence the condition $\rev(\ResdMap(s))\in \acc(s)$ holds (because these transitions can be taken in any order to construct the required word).
 \item Induction step; Following possibilities exist for the state $s$,
 \begin{itemize}
  \item $s$ is a universal state; By construction, there should be states $s_1,\cdots,s_k$ such that $(s,\epsilon,\{s_1,\cdots,s_k\})$ is a transition. By our induction ordering, $s_1,\cdots,s_k$ are smaller than $s$ and hence we apply IH on them to get that $\rev(\ResdMap(s_i))\in \acc(s_i)$ for $i\in\{1\cdots k\}$. However, by the transition rule \ref{tag:compound-assn}, $\ResdMap(s)=\ResdMap(s_1)=\cdots=\ResdMap(s_k)$ and hence $\rev(\ResdMap(s)\in \acc(s_i)$ for $i\in \{1\cdots k\}$. By the definition of $acc(s)$ for a universal state, $\acc(s)$ is intersection of the sets $\acc(s_i)$ for $i\in \{1\cdots k\}$ and hence we get the required result, viz. $\rev(\ResdMap(s))\in \acc(s)$.
  \item $s$ is an existential state; If $s$ is an accepting state then Base case holds here. Consider the case when $s$ is not an accepting state. It should have a successor state $s'$ such that $(s,op,\{s'\})$ is a transition. By transition rule \ref{tag:lit-assn} $\ResdMap(s)=\ResdMap(s').op.\sigma''$ such that $\weakestpre{\subst{\sigma''}{assume}{assert}}{\AssnMap(s)}=\AssnMap(s)$. By transition rule \ref{tag:lit-self-assn}, $s$ will have self loop transitions on all symbols in $\sigma''$(*). Applying IH on $s'$ gives that $\rev(\ResdMap(s'))\in \acc(s')$(\#). Because of the transition $(s,op,\{s'\})$, $op.\acc(s')\subseteq \acc(s)$. This along with (\#) gives us $op.\rev(\ResdMap(s')) \in \acc(s)$(**). Rearranging this and using (*) we get $\rev(\ResdMap(s').op.\sigma'')\in \acc(s)$ or equivalently $\rev(\ResdMap(s))\in \acc(s)$. Hence proved.
 \end{itemize}
\end{itemize}
\end{proof}

\subsection{Proof of Lemma \ref{lem:hoare}}\label{prf:lemhoare}

\begin{proof}
We use induction for this proof. Same as in the previous proof, let us use the following ordering on the states of $\afa{}{}$. For any two states $s$ and $s'$, $s<s'$  if $\len{\ResdMap(s)}<\len{\ResdMap(s')}$ or if lengths are same then $\AssnMap(s)$ is a sub formula of $\AssnMap(s')$. Any two states which are not related by this order, put them in any order to make $<$ as a total order. It is clear that the smallest state in this total order must be one of the accepting state. Now we are ready to proceed by induction using this total order.
\begin{itemize}
 \item Base case, By definition of the accepting state in AFA construction, Point \ref{fin} of Definition \ref{def:afa}, and the self loop transition rule, Rule \ref{tag:lit-self-assn}, we know that for every word $\sigma'\in \acc(s)$, $\weakestpre{\subst{\sigma'}{assume}{assert}}{\AssnMap(s)}=\AssnMap(s)$. Rule \ref{tag:base} of Figure \ref{fig:hmap} sets $\HMap(s)$ same as $\AssnMap(s)$ for such states hence the statement of this lemma follows for the accepting states.
 \item Induction step; we pick a state $s$ such that one of the following holds,
 \begin{enumerate}
 \item \label{prf:univ} $s$ is a universal state;By construction, there should be states $s_1,\cdots,s_k$ such that $(s,\epsilon,\{s_1,\cdots,s_k\})$ is a transition. Let $\sigma$ be a word accepted by $s$ then by the definition of accepting set of words of a universal states, $\sigma$ must be accepted by each of $s_1,\cdots s_k$. By our induction ordering, $s_1,\cdots,s_k$ are smaller than $s$ and hence we apply IH on them to get that $\weakestpre{\subst{\rev(\sigma)}{assume}{assert}}{\AssnMap(s_i)}=\HMap(s_i)$ for $i\in \{1\cdots k\}$. Two cases arise based on whether 
 \begin{itemize}
  \item $\AssnMap(s)$ is a conjunction of $\AssnMap(s_i)$ for $i\in \{1\cdots k\}$; Following Rule \ref{tag:conj} we set $\HMap(s)=\bigwedge_i \HMap(s_i)$ and $\weakestpre{\subst{\rev(\sigma)}{assume}{assert}}{\AssnMap(s)}=\HMap(s)$ then follows from the Property \ref{prop:wp}, using conjunction, of the weakest precondition. 
  \item $\AssnMap(s)$ is a disjunction of $\AssnMap(s_i)$ for $i\in \{1\cdots k\}$; Following Rule \ref{tag:conj} we set $\HMap(s)=\bigwedge_i \HMap(s_i)$ and $\weakestpre{\subst{\rev(\sigma)}{assume}{assert}}{\AssnMap(s)}=\HMap(s)$ then follows from the Property \ref{prop:wp}, using disjunction, of the weakest precondition. 
 \end{itemize}
  \item \label{prf:exist} $s$ is an existential state; If $s$ is an accepting state then the same argument as used in the Base case holds. If $s$ is not an accepting state then the only outgoing transition from $s$ is of the form $(s,op,\{s'\})$, By rule \ref{tag:lit-assn}(*). Now consider a word $\sigma\in \acc(s)$. $\sigma$ must be of the form $\sigma''.op.\sigma'$ where $\weakestpre{\subst{\sigma''}{assume}{assert}}{\AssnMap(s)}=\AssnMap(s)$(*) (because of the self transitions constructed from Rule \ref{tag:lit-self-assn}) and $\sigma'\in \acc(s')$. Therefore, $\weakestpre{\subst{\rev(\sigma)}{assume}{assert}}{\AssnMap(s)}$ =\\
  =$\weakestpre{\subst{\rev(\sigma''.op.\sigma')}{assume}{assert}}{\AssnMap(s)}$\\
  =$\weakestpre{\subst{\rev(\sigma').op.\rev(\sigma'')}{assume}{assert}}{\AssnMap(s)}$\\
  =$\weakestpre{\subst{\rev(\sigma').op}{assume}{assert}}{\AssnMap(s)}$ (using (*))\\
  =$\weakestpre{\subst{\rev(\sigma')}{assume}{assert}}{\weakestpre{\subst{op}{assume}{assert}}{\AssnMap(s)}}$ (using weakest precondition definition)\\
  =$\weakestpre{\subst{\rev(\sigma')}{assume}{assert}}{\AssnMap(s')}$ (using Transition rule \ref{tag:lit-assn})\\
  As $\sigma'\in \acc(s')$ this is same as $\HMap(s')$ by applying IH on $s'$. As $\HMap(s)$ is same as $\HMap(s')$, as done in Rule \ref{tag:prop}, we prove this case as well.
 \end{enumerate}
 \end{itemize}
\end{proof}

\subsection{Proof of Correctness of Transformation-I}\label{prf:modification1hoare}

\begin{lemma}\label{lem:modification1hoare}
 Let $\afa{}{}$ be an automaton constructed from a trace  and a post condition as defined in Definition \ref{def:afa} and further modified by Algorithm \ref{fig:afamodification} then for every state $s$ of this AFA and for every word $\sigma$ accepted by state $s$, $\HMap(s)$ is logically equivalent to $\weakestpre{\subst{\rev(\sigma)}{assume}{assert}}{\AssnMap(s)}$.
\end{lemma}

\begin{proof}
Proof of this lemma is very similar to the proof of Lemma \ref{lem:hoare} given in Appendix \label{prf:lemhoare}. Here we only highlight the changes in the proof. Note that this transformation converts some universal states to existential states. Let $s$ be one such state that was converted from universal to existential state. Let $(s,\epsilon,\{s_1,\cdots,s_k\})$ was the original transition in the AFA which got modified to $(s,\epsilon,\{s_{u_1},\cdots,s_{u_n}\}$ where $s_{u_i}$ are newly created universal states in Line \ref{mod3} of Algorithm \ref{fig:afamodification}. By construction, $\HMap(s_{u_i})$ is unsatisfiable for each of these $s_{u_1},\cdots,s_{u_n}$(*). Let $\sigma$ be a word accepted by $s$ after converting it to existential state. By acceptance conditions, $\sigma$ must be accepted by at least one state, say $s_{u_m}$ in the set $\{s_{u_1},\cdots,s_{u_n}\}$. By IH on $s_{u_m}$ we get $\weakestpre{\subst{\sigma}{assume}{assert}}{\AssnMap(s_{u_m})}=\HMap(s_{u_m})$(**). Further, by construction $\AssnMap(s)$ implies $\AssnMap(s_{u_m})$. This fact, along with the monotonicity property of the weakest precondition, Property \ref{prop:wpsubset}, we get that $\weakestpre{\subst{\sigma}{assume}{assert}}{\AssnMap(s)}$ is unsatisfiable and hence same as $\HMap(s)$.
\end{proof}

\subsection{Proof of Correctness of Transformation-II}\label{prf:modification2hoare}
\begin{lemma}\label{lem:modification2hoare}
 Let $\afa{}{}$ be an automaton constructed from a trace and a post condition as defined in Definition \ref{def:afa} and further modified by adding edges as discussed above then for every state $s$ of this AFA and for every word $\sigma$ accepted by state $s$, $\HMap(s)$ is logically equivalent to $\weakestpre{\subst{\rev(\sigma)}{assume}{assert}}{\AssnMap(s)}$.
\end{lemma}
\begin{proof}
As a result of adding edges in this transformation, we can not use the ordering among states as done for earlier proofs. This is because, now a transition $(s,op,S)$ does not guarantee that the states in the set $S$ are smaller then $s$ and hence it will not be possible to apply IH directly. Therefore in this proof we apply induction on the length of $\sigma'$ accepted by some state $s$.
\begin{itemize}
   \item Induction step; Let $s\in \afa{}{}$ and $\run \in \acc(s)$ such that $\len{\run}=m+1$. Either $s\in \teS$ or $s\in \faS$. If $s\in \teS$ and $\run \in \acc(s)$ then there exists a state $s'$ such that $(s,op, \{s'\})\in \delta$ and $\run'\in \acc(s')$, where $\run=\sigma''.op.\run'$ and $\weakestpre{\subst{\sigma''}{assume}{assert}}{\AssnMap(s)}=\AssnMap(s)$(**). Based on this transition $(s,op,\{s'\})\in \delta$ we have the following sub-cases,
  \begin{itemize}
   \item $(s,op,\{s'\})$ was added by the this transformation virtue of one of the following conditions,
    \begin{itemize}
    \item $\HMap(s)$ and $\HMap(s')$ are unsatisfiable and $\weakestpre{\subst{op}{assume}{assert}}{\AssnMap(s)}\limp \AssnMap(s')$ (Rule \ref{tag:unsatrule}); By IH on $\sigma'$ we have $\weakestpre{\subst{\rev(\run')}{assume}{assert}}{\AssnMap(s')}$ is logically equivalent to $\HMap(s')$. Using Property \ref{prop:wpsubset} (conjunction part) and the assumption $\weakestpre{\subst{op}{assume}{assert}}{\AssnMap(s)}\limp \AssnMap(s')$ we get $\weakestpre{\subst{\rev(\run')}{assume}{assert}}{\weakestpre{op}{\AssnMap(s)}}$ is unsatisfiable and same as $\HMap(s)$. Using (**), $\weakestpre{\subst{\rev(\run')}{assume}{assert}}{\weakestpre{op.\rev(\run'')}{\AssnMap(s)}}$ is unsatisfiable and same as $\HMap(s)$. By replacing $\sigma=\sigma''.op.\sigma'$ we get the required proof.
    \item $\HMap(s)$ and $\HMap(s')$ are valid and $\AssnMap(s')\limp \weakestpre{\subst{op}{assume}{assert}}{\AssnMap(s)}$ (Rule \ref{tag:validrule}); By IH on $\sigma'$ we have $\weakestpre{\subst{\rev(\run')}{assume}{assert}}{\AssnMap(s')}$ is logically equivalent to $\HMap(s')$. Using property \ref{prop:wpsubset} (disjunction part) and the assumption $\AssnMap(s')\limp \weakestpre{\subst{op}{assume}{assert}}{\AssnMap(s)}$ we get $\weakestpre{\subst{\rev}{assume}{assert}(\run')}{\weakestpre{\subst{op}{assume}{assert}}{\AssnMap(s)}}$ is valid and same as $\HMap(s)$. Using (**) and or replacing $\sigma=\sigma''.op.\sigma'$ we get the required result and hence proved.
    \end{itemize}
   \item If this transition was already in $\delta$; we can use the same reasoning as used in the proof of Lemma \ref{lem:hoare} to show that $\weakestpre{\subst{\rev(\run)}{assume}{assert}}{\AssnMap(s)}$ is logically equivalent to $\HMap(s)$
  \end{itemize}
   \item If $s\in \faS$ then similar argument goes as in the proof of Lemma \ref{lem:hoare} because no new transition gets added from these states as a result of this transformation.
  \end{itemize}

\end{proof}

\subsection{Proof of Theorem \ref{lem:termination}}\label{ref:thmtermination}


\begin{proof}
\begin{itemize}
 \item Let us first prove that this algorithm terminates for finite state programs. For finite state programs the number of possible assertions used in the construction of AFA are finite and hence only a finite number of different AFA are possible. It implies the termination of this algorithm.
 \item Following Lemma \ref{lem:hoare} and the fact that $\AssnMap(s_0)=\neg \phi$, every word $\run'$ accepted by this AFA, equivalently written as $\run'\in \acc(s_0)$, satisfies $\weakestpre{\subst{\rev(\run')}{assume}{assert}}{\neg \phi}=\HMap(s_0)$(*). By Lemma \ref{lem:runinafa} and the fact that $\ResdMap(s_0)=\run$ we get $\rev(\run)\in \acc(s_0)$(**). Combining (**) and (*), we get $\weakestpre{\subst{\rev(\rev(\run))}{assume}{assert}}{\neg \phi}=\HMap(s_0)$ or equivalently $\weakestpre{\subst{\run}{assume}{assert}}{\neg \phi}=\HMap(s_0)$. 
 \begin{itemize}
  \item If $\mathcal{I} \lland \HMap(s_0)$ is satisfiable (Line \ref{finalcounter}) then $\mathcal{I}\lland  \weakestpre{\subst{\run}{assume}{assert}}{\neg \phi}$ is satisfiable as well. Following Lemma \ref{lem:assumeassrtlhs2} we got a valid error trace which is returned in Line \ref{finalret1}.
  \item If $\mathcal{I} \lland \HMap(s_0)$ is unsatisfiable then by Lemma \ref{lem:assumeassrtrhs2} this trace is provably correct. Now we apply transformations of Section \ref{subsec:trans} on the AFA to increase the set of words accepted by it. The final AFA is then reversed and subtracted from the set of executions seen so far. Lemma \ref{lem:hoare} ensures that for all such words $\run'$ the condition $\mathcal{I}\nRightarrow \weakestpre{\run'}{\neg \phi}$ holds and therefore none of them violate $\phi$ starting from the initial state. Therefore in every iteration only correct set of executions are being removed from the set of all executions. Therefore when this loop terminates then all the executions have been proved as correct. 
  \end{itemize}
\end{itemize}
\end{proof}

\end{document}